\newif\ifreviewmode
\reviewmodefalse

\ifreviewmode
    \documentclass[12pt,journal,onecolumn]{IEEEtran}
    \newcommand{\mywidth}{0.7\columnwidth}

    \usepackage{xparse}
    \NewDocumentEnvironment{mymultline*}{+b}
        {
            \begin{equation*}
            \begingroup
            \renewcommand{\\}{\space}
            #1
            \endgroup
            \end{equation*}
        }
        {}
  
\else
    \documentclass[lettersize,journal]{IEEEtran}
    \newcommand{\mywidth}{\columnwidth}

    \usepackage{xparse}
    \NewDocumentEnvironment{mymultline*}{+b}
        {
            \begin{multline*}
            #1
            \end{multline*}
        }
        {}
\fi

\usepackage{amsmath,amsfonts, mathtools, amsthm}
\usepackage{algorithmic}
\usepackage{algorithm}
\usepackage[dvipsnames]{xcolor}
\usepackage{array}
\usepackage[caption=false,font=normalsize,labelfont=sf,textfont=sf]{subfig}
\usepackage{textcomp}
\usepackage{stfloats}
\usepackage{url}
\usepackage{verbatim}
\usepackage{graphicx}
\usepackage{cite}
\newcommand{\citeasnoun}[1]{Ref.~\citenum{#1}}
\usepackage{hyperref}
\usepackage[textwidth=0.55in,textsize=scriptsize]{todonotes}
\usepackage{cleveref}

\newtheorem{definition}{Definition}
\newtheorem{theorem}{Theorem}
\newtheorem{corollary}{Corollary}
\newtheorem{lemma}{Lemma}

\newtheorem{proposition}{Proposition}
\newtheorem{assumption}{Assumption}

\newcommand{\N}{\mathbb{N}}
\newcommand{\Z}{\mathbb{Z}}

\newcommand{\R}{\mathbb{R}}

\newcommand{\Law}{\mathrm{Law}}
\newcommand{\E}{\mathbb{E}}
\newcommand{\Var}{\mathrm{Var}}
\newcommand{\Ind}[1]{\mathbf{1}_{\{#1\}}}

\DeclareMathOperator{\Ber}{Ber}

\DeclareMathOperator{\Unif}{Unif}
\DeclareMathOperator{\BetaDist}{Beta}
\newcommand{\Normal}{\mathcal{N}}
\newcommand{\Tr}{\operatorname{tr}}
\newcommand{\Diag}{\operatorname{diag}}

\begin{document}

\title{Fundamental Bounds and Efficient Estimation for Dead-Time-Constrained Event Detection, with Application to Single-Photon Lidar}
\author{Frederic J. N. Jorgensen and Steven G. Johnson%
\thanks{This work was supported in part by the Singapore-MIT Alliance for Research and Technology (SMART) Wafer-scale Integrated Sensing Devices based on Optoelectronic Metasurfaces (WISDOM) Interdisciplinary Research Group, and by a grant from the Simons Foundation.}%
\thanks{F. J. N. Jorgensen and S. G. Johnson are with the Department of Mathematics, Massachusetts Institute of Technology, Cambridge, MA 02139 USA. Corresponding author: F. J. N. Jorgensen; e-mail: fjorgen@mit.edu.}

}



\maketitle

\begin{abstract}
We develop an asymptotic statistical theory for parameter estimation from a class of non-i.i.d.\ periodic binary event-detection processes subject to nonparalyzable dead time and gating, which we call ``dead-time event detection'' (DED) processes. Such processes arise in single-photon lidar, fluorescence lifetime imaging, X-ray astronomy, and particle or radiation flux  measurements in nuclear physics, where each detection renders the radiation/particle detector inactive for a recovery interval.  Our theory quantifies how dead time and gating affect the fundamental lower bounds of estimation and identifies practical estimators that attain these bounds. First, we identify a sufficient statistic, showing in particular that activation counts can carry statistically useful information discarded by conventional histogramming hardware. We then prove local asymptotic normality and derive the corresponding Fisher-information rate, thereby obtaining fundamental lower bounds for estimation from DED processes. We prove that the maximum likelihood estimator (MLE), widely used in DED applications, attains these lower bounds. Since computing the MLE typically requires solving a nonconvex optimization problem, we also propose Le Cam one-step estimators, which attain the same asymptotic bounds with only a single local correction rather than iterative optimization.  We illustrate the validity of our asymptotic theory and the practical usefulness of one-step estimators  through the example of single-photon lidar in both simulations and real-data experiments. 
\end{abstract}

\section{Introduction}
We develop an asymptotic theory of parameter estimation for periodic binary detection processes subject to dead time. Our framework encompasses a wide range of particle and radiation measurement problems in science and engineering (see \Cref{section:examples_problem}), and we demonstrate its usefulness through an application to lidar (light detection and ranging)~\cite{fiocco1963detection,dong2017lidar,  kirmani2014first}, including an example with real (experimental) data.   Concretely, the systems we study conduct periodic statistical experiments, each with a binary outcome indicating whether a particle or radiation quantum was detected. From the observed outcomes, one wishes to estimate a physical parameter of interest, such as a pulse time-of-flight in lidar~\cite{dong2017lidar} or a lifetime in fluorescence microscopy~\cite{becker2004fluorescence, becker2012fluorescence}. A common property of such systems is that, whenever a detection event occurs, the underlying device enters a ``dead-time'' period of duration $t_d$ during which no further outcomes can be registered. Such a dead-time interval arises in practice when, for example, particle or radiation detectors are operated in Geiger mode and require a short recovery upon detection before returning to their sensitive state~\cite{muller1973dead, buller2009single, isbaner2016dead, rapp2019dead, oconnor1984time}. To mitigate its effect, experiments are often deliberately activated only during selected time bins, a design choice called \emph{gating}, which trades off conducting the most informative experiments  against maximizing the number of experiments~\cite{buttafava2014time, buttafava2015non, wang1991two, tosi2011fast, reilly2014high, gupta2019asynchronous, po2022adaptive, isbaner2016dead, rapp2019dead}.  Together, dead time and gating introduce a systematic discrepancy between the observed and the underlying experimental statistics. In many regimes of practical importance this distortion is statistically significant, and a substantial literature across application domains has proposed empirical correction schemes~\cite{coates1968correction, muller1973dead, lee2000new,
	taguchi2011modeling, isbaner2016dead, pediredla2018signal,
	gupta2019photon, rapp2019dead, huppenkothen2022accurate,
	larsen2009simple, kitichotkul2025free, wang2026synchronous} (see \Cref{section:existing_approaches}).  Previous statistical treatments, however, do not account for the non-i.i.d.\ stochastic process resulting from dead time and gating (see \Cref{section:related_work}), leaving open both the fundamental lower bounds for this class of problems (which can guide the principled design of estimators and gating schemes) and the efficiency of commonly used estimators. This paper fills that gap.

In \Cref{section:model_formulation}, we precisely define periodic ``\textbf{d}ead-time \textbf{e}vent \textbf{d}etection (\textbf{DED})'' processes to formalize the general class of problems described above, depicted schematically in \Cref{fig:ded_process}. We identify a sufficient statistic of this process, showing that activation counts can carry information about the physical parameters of interest that hardware-side histogramming typically discards in practice~\cite{rapp2019dead, gupta2019photon, po2022adaptive, isbaner2016dead}. This observation can inform  future hardware histogramming design. In \Cref{section:asymptotic_theory}, we develop an asymptotic theory of estimation for DED processes, characterizing the fundamental lower bounds of estimation from dead-time--constrained observations. Using the local asymptotic normality (LAN) framework, we derive the appropriate notion of Fisher information for this non-i.i.d.\ setting and obtain matching lower bounds through H\'ajek convolution and local asymptotic minimax theorems. A key consequence of our Fisher-information-rate formula is that the gating rule, even when fully history-dependent, affects asymptotic information only through its limiting phasewise sampling frequencies and, weighted by these frequencies, each sampled phase contributes the usual Bernoulli information. 
In \Cref{section:efficiency}, we prove that the maximum likelihood estimator (MLE), of which several implementations are commonly used in practice \cite{coates1968correction, pediredla2018signal, rapp2019dead, kitichotkul2025free, maus2001experimental, chessel2013maximum, bajzer1991maximum, wang2026synchronous
}, is asymptotically efficient. While prior work has only treated specific gating rules, we further show how to compute the MLE for general gating rules in the dead-time setting. 
As an alternative to the MLE for estimation from DED processes, we propose a Le~Cam one-step update, yielding a class of estimators known as one-step estimators. 
This class of estimators decouples global parameter localization from locally efficient estimation, where a pilot estimator handles the former and can be designed for stability, robustness, or computational simplicity, while a single update step provides the efficient estimation, which we establish formally. 
This avoids the pitfalls of nonconvex MLE optimization while retaining its asymptotic efficiency. 
To illustrate the practical usefulness of these ideas, in \Cref{section:lidar} we apply the theory to single-photon lidar. We verify that the widely used MLE is asymptotically efficient in this setting, but demonstrate empirically on simulated and real data that one-step estimators (built from a new robust matched-filter pilot we introduce) match the MLE's efficiency while avoiding both its sensitivity to initialization and its degraded performance under model misspecification.  For the specific task of ranging, we apply one-step updates to several estimators from the single-photon lidar literature \cite{gupta2019photon, kitichotkul2025free, pawlikowska2017single, mccarthy2013kilometer}  and show that these corrections substantially improve their performance. We also find that our robust one-step estimator (with only log-linear cost) is competitive with, and arguably outperforms, these corrected literature baselines.

\begin{figure*}[!tb]
    \centering
    \includegraphics[width=\textwidth]{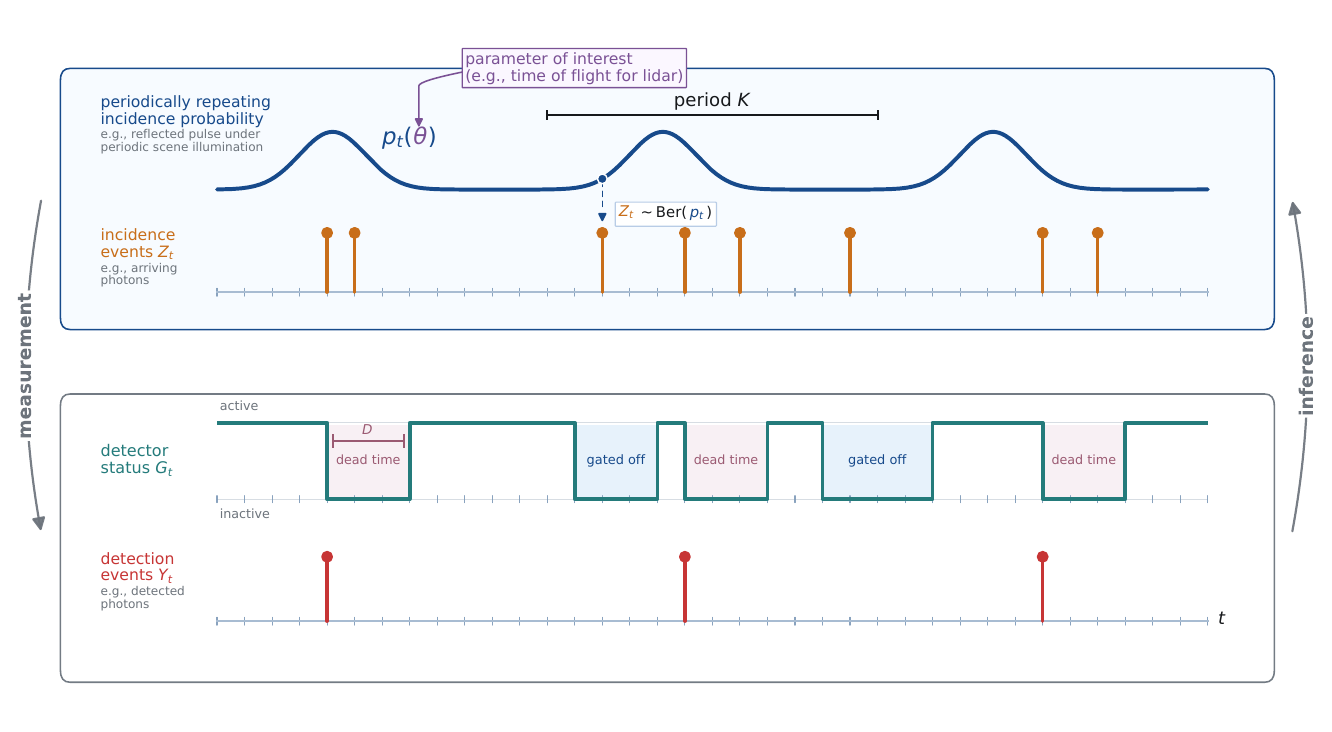}
    \caption{Schematic of a dead-time event detection (DED) process. The periodic incidence probability \(p_t(\theta)\) generates latent Bernoulli incidence events \(Z_t\). The detector status \(G_t\) is determined by the gating scheme and dead time. Only incidences occurring while the detector is active are recorded as detections \(Y_t=G_t Z_t\), and dead time occurs after every $Y_t=1$ detection. The goal is to infer \(\theta\) from the detector data \(\{(G_t,Y_t)\}_{t\geq0}\). In lidar, for example, \(p_t(\theta)\) is the photon incidence probability resulting from a repeating return-pulse profile, with \(\theta\) including the time-of-flight to be estimated from the data.}
    \label{fig:ded_process}
    \end{figure*}

\subsection{Examples of event detection with dead time}
\label{section:examples_problem}
DED estimation problems arise in a wide range of physical settings, several of which we briefly describe below.  In these applications, the binary event stream records whether at least one particle or radiation quantum was detected in the corresponding time interval.
\paragraph{Single-photon imaging at optical wavelengths}

As single-photon detectors have become more capable and less expensive~\cite{buller2009single}, time-correlated single-photon counting (TCSPC) has become a widely used technique in a range of imaging applications. In TCSPC, a system is illuminated by a periodic pulse and the statistics of photon arrivals in response to each pulse are used to infer physical parameters of interest~\cite{oconnor1984time}. Three representative examples are the following. In lidar~\cite{dong2017lidar, kirmani2014first}, each event records the detection of at least one returning photon from an illuminated scene point, and the resulting arrival-time histogram is used to estimate the scene depth and reflectivity at that point. In fluorescence lifetime imaging (FLIM)~\cite{becker2004fluorescence, becker2012fluorescence}, the recorded events are fluorescence photons whose arrival-rate function decays exponentially after each excitation pulse, and the parameter of interest is an exponential decay constant that corresponds to the fluorescence lifetime of a molecular species. In non-line-of-sight (NLOS) imaging~\cite{drost2015dead, rapp2020seeing, faccio2020non}, the detector is scanned over a diffuse relay surface and the photon-arrival rate at each scan position is used to reconstruct a scene hidden around a corner. In all three settings, dead time can systematically distort the observed arrival-time histogram relative to the underlying photon-arrival distribution~\cite{rapp2019dead, li2017influence, xu2017signal}.
\paragraph{Single-photon detection at non-optical wavelengths}
The same model applies to single-photon detectors operating outside
the optical spectrum. A prominent example is the estimation of a periodic
photon rate using X-ray detectors in astronomy, where the periodicity
arises from physical processes such as pulsar rotation
\cite{weisskopf2002overview, tuo2024revisiting} or orbiting exoplanets
\cite{poppenhaeger2013transit}.

\paragraph{Particle detection}
In particle and nuclear physics, and in astroparticle physics, an
important application is the estimation of a constant particle or
radiation flux from individual detection events
\cite{knoll2010radiation, gesualdi2022estimation, aab2016pierre,
reines1954detection, pal2012some, lee2000new}.

\paragraph{Closely related problems}
Several related problems share structural features but violate at least one of the hypotheses of our setting (single detector, event-based, dead-time, periodic). Medical imaging modalities such as positron emission tomography (PET) and computed tomography (CT) employ photon and particle counters described by the same single-detector model~\cite{vicente2013improved, taguchi2011modeling}, but use \emph{arrays} of detectors whose outputs are jointly processed to reconstruct an activity or attenuation distribution inside a body, yielding a multidimensional generalization of our problem.  A similar  multi-detector setup appears in recent work on vibration sensing via  frequency-entangled two-photon interference~\cite{lualdi2026quantum, lualdi2025fast}, where photon pairs from a  Poisson source are routed through an interferometer and their  coincidence and anti-coincidence rates at two single-photon detectors  encode a time-varying optical delay, which is the parameter to be  inferred.  Another related problem arises in the neuroscience literature on neural
encoding, where a spike-train response (e.g.~the action potentials of retinal
ganglion cells) is modeled parametrically as a function of a known stimulus,
with the goal of estimating the parameters of this model.  These models extend our setting in several ways: {refractory} times, the analogue of dead time in this context, are often nonconstant, and the measured spike-train response  may
exhibit additional forms of history dependence, referred to in this literature
as burstiness and adaptation. Moreover, they frequently, though not always, use
aperiodic or stochastic stimuli rather than the periodic setting considered here \cite{pillow2008spatio, keat2001predicting, mcintosh2016deep,berry1997refractoriness}. 
A further related direction concerns {refractory}  periods in models of recurrent event data in medical studies, which are mathematically equivalent to dead time. These  typically lack the periodic rate function central to our setting~\cite{jahn2015simulating}. 
Refractory-like constraints also  arise in queueing applications such as databases and aircraft landings~\cite{galliher1958nonstationary}, but there the constraint does not  restrict what the observer can record and so does not induce the  estimation problem studied here.
Finally, variability in electromagnetic radiation emitted from accreting galactic black holes and similar astrophysical sources gives rise to a closely related inference problem, but with a stochastic rather than periodic rate function~\cite{huppenkothen2022accurate}.

\subsection{Existing approaches and their limitations}
\label{section:existing_approaches}
A variety of strategies have been proposed to mitigate dead-time--induced  bias in parameter estimation. Although most of these ideas have been developed in  the context of lidar and single-photon avalanche diode (SPAD) imaging,  the underlying ideas are general and apply to any  dead-time--constrained detection system. Throughout, we use the term \emph{detector} for the device that reads out the binary event of interest, refer to it as \emph{active} when it is able to register an event and \emph{inactive} otherwise. 

In low-event-rate regimes, dead-time effects are small and are often
simply neglected at the inference stage~\cite{oconnor1984time,valeur2013molecular,
kirmani2014first}. Alternatively, computational correction methods can
compensate for the dead-time--induced distortion at the inference
stage~\cite{coates1968correction, muller1973dead, lee2000new,xu2017signal,hou2021full,
taguchi2011modeling, isbaner2016dead, pediredla2018signal,
gupta2019photon, rapp2019dead, huppenkothen2022accurate,
larsen2009simple, kitichotkul2025free}.

A complementary hardware-level approach is \emph{gating}, in which the  detector is deliberately activated and deactivated at strategically  chosen times to counter pile-up and, potentially, improve temporal  resolution and reduce reconstruction error. Representative strategies  include activating the detector only over a sub-interval in which  signal events are expected~\cite{buttafava2014time, buttafava2015non},  and shifting the active window across repeated measurement periods~\cite{wang1991two, tosi2011fast, reilly2014high, pifferi2008time, gupta2019asynchronous}.  \emph{Adaptive} gating uses information from previous measurements to  decide when to open the active window~\cite{po2022adaptive}. Finally,  event-driven or \emph{free-running} acquisition, in which the detector  is active whenever it is not in a detection-induced dead-time period,  has been studied as an alternative to explicit gating~\cite{isbaner2016dead, rapp2019dead}.

\subsection{Related work}
\label{section:related_work}
Despite their empirical success, existing inference methods and gating schemes are motivated by heuristic considerations rather than derived from an explicit statistical framework. Several questions have been left open.  To begin with, what notion of Fisher information governs the asymptotic lower bounds for estimation in this setting and do commonly used estimators achieve them? The answers to these questions will also inform how gating schemes should be designed under physical constraints (\Cref{section:conclusion}). The asymptotic theory of estimation for periodic binary processes \emph{without dead time}, which reduce to product-binomial models, is well-studied~\cite{berk1972consistency, efron2022exponential, van2000asymptotic}. To the best of our knowledge, however, no analogous theory exists for the dead-time--constrained setting considered here. The dead-time constraint and the gating feedback induce a dependency between successive observations, ruling out a direct application of classical i.i.d.\ asymptotics. We instead apply local asymptotic normality (LAN) theory~\cite{le2000asymptotics, van2000asymptotic}, which accommodates this dependence structure while still yielding sharp lower bounds and a well-defined notion of asymptotic efficiency.

A related but separate line of work derives lower bounds under the restrictive assumption of \emph{unbiased} estimators, rarely satisfied in practice, either neglecting dead-time effects~\cite{kitichotkul2023role, pediredla2018signal, gupta2019asynchronous, lu2013adaptive} or restricting to specific forms of gating~\cite{rebafka2011information, wu2025performance}, whereas our LAN-based bounds apply to possibly biased estimators with arbitrary gating.

\subsection*{Notation}
We write $\N=\{0,1,2,\ldots\}$. For $a,b\in\R$, write $a\wedge b\coloneq   \min\{a,b\}$.  For $d\in\N$, $\R^d$ denotes $d$-dimensional Euclidean space. We write $\|\cdot\|_2$ for the Euclidean norm and $\|\cdot\|_{\text{F}}$ for the Frobenius norm. For a matrix $A$, $\Tr(A)$ denotes its trace and $\Diag(a_1,\ldots,a_d)$ the diagonal matrix with diagonal entries $a_1,\ldots,a_d$. For a finite sequence $x_0,\ldots,x_{T-1}$, we use the shorthand $x_{0:T-1}\coloneq (x_0,\ldots,x_{T-1})$. The law of a random variable $X$ is denoted by $\Law(X)$, and expectation is denoted by $\E$. For probability measures $P\ll Q$, \(D_{\mathrm{KL}}(P\|Q)\coloneq \E_P[\log(dP/dQ)]\) denotes the Kullback--Leibler divergence. We use $\Rightarrow$ for convergence in distribution and $\xrightarrow{P}$ for convergence in probability under $P$. $X_T = o_P(1)$ if $X_T \xrightarrow{P} 0$, and $X_T = O_P(1)$ if $X_T$ is bounded in probability.

\section{Model formulation}
\label{section:model_formulation}

We formalize the task of inferring a physical parameter $\theta\in\Theta\subseteq\R^d$ from a periodically repeating sequence of Bernoulli experiments, representing particle or radiation detection events, in which each positive event forces the detector into a \emph{dead-time} period of duration $t_d$ during which no further experiments can be conducted. A second feature we want to capture is \emph{gating}: subject to the dead-time constraint, the experimenter is free to choose \emph{whether} to run an experiment at any given time. The design of gating schemes is important in practical applications because skipping an uninformative experiment preserves the detector's sensitivity for later, potentially more informative, time intervals~\cite{buttafava2014time, buttafava2015non, wang1991two, tosi2011fast, reilly2014high, gupta2019asynchronous, po2022adaptive, isbaner2016dead, pifferi2008time, rapp2019dead}.  

The remainder of this section is organized as follows. \Cref{subsection:observation_model} defines our observation model and the class of stochastic processes we study, which we call \textbf{d}ead-time \textbf{e}vent \textbf{d}etection (DED) processes. \Cref{subsection:physical_interpretation} connects this model to concrete applications. \Cref{subsection:likelihood_suff} establishes useful properties of the likelihood, including a sufficient statistic, and connects these results to hardware data compression in bandwidth-limited systems.

\subsection{Observation model}
\label{subsection:observation_model}
We work in discrete time, $t\in\{0,1,2,\ldots\}=\N$, with period $K$ a positive integer (see \Cref{subsubsection:implicit_modeling_assumptions} for a discussion of this choice). We write \(D\in\N\) for the dead-time length in discrete time bins. The underlying physical process has $K$-periodic event probabilities
$$
p_t(\theta)=p_{t+K}(\theta)\in[0,1] \quad\text{for all }t\in \N \, ,
$$
depending on the physical parameter  $\theta\in\Theta\subseteq\R^d$ that the experimenter wishes to estimate. The corresponding \emph{physical} Bernoulli experiments are sampled independently as
\begin{equation}
\label{eq:bernoulli_experiments}
Z_t \sim \Ber(p_t(\theta)),
\qquad t\in \N.
\end{equation}
In practice, these events correspond to the impingement of at least one particle/radiation quantum. Because of dead time and gating, the experimenter does \emph{not} actually observe the entire sequence $\{Z_t\}_{t\ge 0}$. Instead, the experimental data consist of a stochastic process $\{(G_t,Y_t)\}_{t\ge 0}$ in which $G_t\in\{0,1\}$ indicates whether an experiment is observed at time $t$, and the observed outcome is
\begin{equation}
\label{eq:observed_experiment}
Y_t = G_t Z_t.
\end{equation}
Thus, $Y_t$ coincides with $Z_t$ whenever the detector is active ($G_t=1$), while it is identically zero and uninformative about $Z_t$ whenever the detector is inactive ($G_t=0$).

To fully define the law of $\{(G_t,Y_t)\}_{t\ge 0}$, we must specify how the random variables $G_t$ are chosen. We impose three physical conditions:
\paragraph{Dead-time constraint} A detected event at time $t$ ($Y_t = 1$) prevents any experimental outcome from being observed over the following $D$ time steps:
\begin{equation}
\label{eq:dead_time}
Y_t = 1  \Longrightarrow  G_s = 0
\text{ for } s \in \{t+1, t+2, \dots, t+D \}
\end{equation}
almost surely (a.s.). We refer to stochastic processes $\{(G_t,Y_t)\}_{t\ge 0}$ satisfying \eqref{eq:dead_time} as \emph{dead-time--constrained}. In Geiger-mode detection experiments, this reflects the recovery time required by the detector after a detection~\cite{muller1973dead}.

\paragraph{Gating causality} Subject to the dead-time constraint, the experimenter is free to choose at any time whether to run an experiment, and a substantial literature has developed proposing schemes for doing so~\cite{buttafava2014time, buttafava2015non, wang1991two, tosi2011fast, reilly2014high, gupta2019asynchronous, po2022adaptive, isbaner2016dead, rapp2019dead}. To include any such scheme while ruling out an unphysical dependency on future information, we require that for every $t$ there exist $U_t\sim\Unif([0,1])$, independent of the physical event sequence $\{Z_s\}_{s\ge0}$ and of the past observations $\{(G_s,Y_s)\}_{s=0}^{t-1}$, and a measurable map $\Phi_t:\{0,1\}^{2t}\times[0,1]\to\{0,1\}$ such that, a.s.,
\begin{equation}
\label{eq:causality}
G_t = \Phi_t\bigl(G_0,Y_0,\dots,G_{t-1},Y_{t-1},U_t\bigr).
\end{equation}
Equivalently, $G_t$ is drawn from a causal Markov kernel $\pi_t(\,\cdot\mid G_0,Y_0,\dots,G_{t-1},Y_{t-1})$ defined as the law of $\Phi_t$. Subject to \eqref{eq:dead_time}, the choice of the family $\{\Phi_t\}_{t\ge 0}$ constitutes the experimenter's design freedom. In the following examples, we use the convention that \(Y_t=0\) for \(t<0\). The following two choices are widely used in practice:
\begin{itemize}
	\item The \emph{free-running} scheme~\cite{rapp2019dead, gupta2019asynchronous}, $\Phi_t^{\mathrm{fr}} = 1 - \sum_{s=1}^{D} Y_{t-s}$, which observes an experiment whenever allowed by the dead-time constraint.
\item The \emph{synchronous} scheme~\cite{buttafava2014time, buttafava2015non}, in which the detector is activated at the beginning of each period unless still in dead time, and deactivated after the first detection within the period. Writing \(m_t\coloneq t\bmod K\), this rule is $\Phi_t^{\mathrm{syn}} = \left(1-\sum_{s=m_t+1}^{m_t+D}Y_{t-s}\right) \left(1-\sum_{s=1}^{m_t}Y_{t-s}\right),$ with the second sum interpreted as empty when \(m_t=0\). 
\end{itemize}
The trade-off underlying gating design is that not running an experiment reduces the amount of observed data, but can preserve the detector's sensitivity for subsequent, potentially more informative, time intervals. The statistical framework developed in this paper provides a basis for principled comparison and design of gating schemes \(\{\Phi_t\}_{t\ge 0}\) in future work (\Cref{section:conclusion}).
\paragraph{Gating-frequency convergence} We require that the empirical frequency of detector activations along each phase stabilizes, meaning that for every $\theta\in\Theta$, there exists a limiting ``gating-frequency'' vector $\gamma(\theta)\in[0,1]^K$ such that for each $r=0,\dots,K-1$,
\begin{equation}
\label{eq:phase_convergent}
\frac{1}{L}\sum_{\ell=0}^{L-1} G_{r+\ell K}
\xrightarrow{P_\theta}
\gamma_r(\theta)
\qquad\text{as } L\to\infty.
\end{equation}
When the choice of $\theta$ is clear from context, we write $\gamma$ for $\gamma(\theta)$.  Gating-frequency convergence holds for all practical examples we know of and is verified for the free-running and synchronous rules in \Cref{app:gating_frequency_convergence_acquisition}.  It could likely be relaxed at the cost of more involved proofs but little additional statistical insight, and schemes that rely on such a relaxation are likely undesirable in practice (see \Cref{subsubsection:circumventing_gating_frequency_convergence}).

The following definition collects these requirements and names the class of processes studied in this paper.
\begin{definition}
\label{def:DED}
A stochastic process $\{(G_t, Y_t)\}_{t\geq 0}$ with $G_t, Y_t \in \{0, 1\}$ is called a \textbf{d}ead-time \textbf{e}vent \textbf{d}etection (DED) process if it is causal as in \eqref{eq:causality}, gating-frequency convergent as in \eqref{eq:phase_convergent}, dead-time--constrained as in \eqref{eq:dead_time}, and satisfies the observation relations \eqref{eq:bernoulli_experiments}--\eqref{eq:observed_experiment}.
\end{definition}
The law of a DED process is fully specified by the periodic event probabilities $p_r(\theta)$ and the gating rules $\{\Phi_t\}_{t\ge 0}$ through a simple recursion.  The latter determine the conditional law of $G_t$ given the past, and \eqref{eq:observed_experiment} then determines $Y_t$ given $G_t$.

\subsection{Physical examples}
\label{subsection:physical_interpretation}
We now connect the abstract model above to concrete physical problems, using single-photon detection as the primary example. Detection of massive particles exhibits the same Poissonian statistics and admits the same argument. In a single-photon counting experiment, $p_t(\theta)$ is the probability that at least one photon arrives in the time bin $t$. Photon arrival processes are standardly modeled semiclassically as a Poisson point process~\cite{goodman2015statistical, snyder2012random}, with a nonnegative photon rate $\lambda_t(\theta)\ge 0$ that is proportional to the classical light intensity. This model is also widely used in the single-photon detection literature~\cite{rapp2019dead,xu2017signal, gupta2019photon}. Because a single-photon detector can only register the binary event ``at least one photon arrived'' (and not the photon number)~\cite{buller2009single}, the detection probability in bin $t$ is
\begin{equation}
\label{eq:reparameterization}
p_t(\theta)  =  1 - \exp\bigl(-\lambda_t(\theta)\bigr).
\end{equation}
The dependence of $p_t$ on $\theta$ through the photon rate $\lambda_t(\theta)$ is what makes single-photon detectors useful for parameter $\theta$ estimation. We list three important examples.
\begin{itemize}
    \item \textbf{Lidar:} The  photon rate $\lambda_t(\theta)$ is a pulse-shaped function of time whose position and amplitude encode, respectively, the round-trip time-of-flight to a scene point and the scene reflectivity, on top of a constant ambient background. The parameter $\theta$ collects the time-of-flight, return amplitude, and background level~\cite{buller2009single,kirmani2014first}. A detailed treatment is given in \Cref{section:lidar}.
    \item \textbf{Fluorescence lifetime imaging (FLIM):} Following a short excitation pulse, fluorescent molecules emit photons at an exponentially decaying rate $\lambda_t(\theta)$.   The parameter of interest $\theta$  in this application  is the exponential decay constant called ``fluorescence lifetime,'' which serves as an important imaging modality in biological experiments~\cite{becker2004fluorescence,wang1991two}.
    \item \textbf{X-ray timing in astronomy:} For a periodically varying astrophysical source such as a pulsar or a transiting-exoplanet system, the  photon rate $\lambda_t(\theta)$ is a periodic pulse profile whose parameters (e.g.~the orbital period, pulse shape, and pulsed fraction) are the inference targets~\cite{weisskopf2002overview, tuo2024revisiting, poppenhaeger2013transit}.
\end{itemize}
We remark that, in practice, the detection probabilities in \eqref{eq:reparameterization} also depend on the photon detector efficiency~\cite{goodman2015statistical}, which acts multiplicatively on the photon rate. Without loss of generality, this factor can be absorbed directly into the rate function $\lambda_t(\theta)$, and we omit it from the notation.

\subsection{Likelihood, sufficient statistic, and Rao--Blackwellization}
\label{subsection:likelihood_suff}
We now derive the likelihood of a DED process up to a finite time horizon $T\in\N$ and identify a sufficient statistic for~$\theta$. Writing $X_t \coloneq  (G_t, Y_t)$, with a feasible realized value $x_t=(g_t,y_t)$, denoting the finite-horizon law by \(P_{\theta,T}\coloneq \Law(X_{0:T-1})\), writing \(\E_{\theta,T}\) for expectation with respect to \(P_{\theta,T}\) and suppressing \(T\) as in \(\E_\theta\) when the horizon is clear, the chain rule combined with \eqref{eq:observed_experiment}, and the definition of $\pi_t$ yields the factorization
\begin{subequations}\label{eq:factorization}
\begin{align}
P_{\theta,T}(X_{0:T-1} = x_{0:T-1})  = \prod_{t=0}^{T-1}\pi_t(g_t\mid x_{0:t-1})   \\ 
 \cdot \underbrace{\prod_{t=0}^{T-1}
\bigl[p_t(\theta)^{y_t} (1-p_t(\theta))^{1-y_t}\bigr]^{g_t}}_{\theta\text{-dependent}},
\end{align}
\end{subequations}
where only the second factor depends on $\theta$, and where we use the convention that the bracketed term equals $1$ when $g_t = 0$. Define the phasewise aggregate counts
\begin{subequations}
\label{eq:phase_counts}
\begin{align}
N_r(T)& \coloneq  \sum_{\substack{0\le t<T\\ t\bmod K=r}} G_t,
\\
S_r(T)& \coloneq  \sum_{\substack{0\le t<T\\ t\bmod K=r}} Y_t,
\end{align}
\end{subequations}
for $r = 0, \dots, K-1,$
i.e.~the number of active bins and the number of detections within phase $r$ up to horizon $T$. Substituting \eqref{eq:phase_counts} into \eqref{eq:factorization} gives
\begin{subequations}
\label{eq:phase_loglikelihood_intro}
\begin{align}
\log P_{\theta,T}&(X_{0:T-1} = x_{0:T-1})
 =\sum_{r=0}^{K-1}\Bigl[
S_r(T)\log p_r(\theta)\\
 & + \bigl(N_r(T) - S_r(T)\bigr)\log\bigl(1-p_r(\theta)\bigr)
\Bigr]\\
&
+ C(x_{0:T-1}),
\end{align}
\end{subequations}
where $C(x_{0:T-1})$ collects $\theta$-independent terms. The following proposition is immediate from the form of  \eqref{eq:phase_loglikelihood_intro} and the Neyman--Fisher factorization theorem; see, e.g., \citeasnoun{casella2024statistical}.
\begin{proposition}
\label{prop:suff_active}
For every DED process and every $T\in\N$, the vector of phasewise counts
$$
H_T  \coloneq   \bigl(N_r(T), S_r(T)\bigr)_{r=0,\dots,K-1}
$$
is sufficient for $\theta$, meaning that the conditional distribution of the full data $X_{0:T-1}$ given $H_T$ does not depend on $\theta$. 
\end{proposition}
Importantly, this implies that any estimator $\hat\theta_T$ can be replaced by the conditional expectation $\E[\hat\theta_T \mid H_T]$,  which can be computed without knowing $\theta$ by sufficiency, without increasing the mean squared error at any $\theta\in\Theta$.   Therefore, without loss of optimality we may restrict attention to estimators that depend on the data only through $H_T$.  This principle is also known more generally as the Rao--Blackwell theorem~\cite{blackwell1947conditional}.   

This sufficiency result has direct practical implications for data compression, which is an important problem in modern detection systems such as high-refresh-rate lidar~\cite{ingle2023count}. \Cref{prop:suff_active}  tells us that, rather than the full binary stream $\{(G_t, Y_t)\}_{t=0}^{T-1}$,  only the $2K$ phasewise counts $(N_r(T), S_r(T))_{r=0}^{K-1}$ need to be stored or transmitted from detector hardware to the inference device. This explicit identification of a sufficient statistic provides a principled basis for ongoing work on hardware data compression in bandwidth-limited systems.
Classical histogramming retains only $S_r(T)$~\cite{rapp2019dead, gupta2019photon,xu2017signal, po2022adaptive, isbaner2016dead}. Depending on the gating scheme used, this \textit{can} be a sufficient statistic. For example, in the free-running case with time horizon $T = L K$, since $G_t = 1 - \sum_{s=1}^{D} Y_{t-s}$, one has
\begin{equation}
\label{eq:Nr_from_Sr_freerunning}
N_r(T) =  L -  \sum_{s=1}^{D} S_{(r-s)\bmod K}(T) + O(1)
\end{equation}
where the $O(1)$-term  is due to boundary effects and negligible as $L$ becomes large. 
For more general gating schemes where no closed-form expression of the type \eqref{eq:Nr_from_Sr_freerunning}  exists (e.g.~adaptive schemes as in~\citeasnoun{po2022adaptive}),  the counts $N_r(T)$ carry additional statistical information that is lost in $S_r(T)$-histogramming. To our knowledge, this point has not been made explicit in the existing literature. Recent work has empirically proposed retaining $N_r(T)$ alongside $S_r(T)$ for use in a computational correction~\cite{rech2023toward, fratta2024near}. Under the free-running scheme used in that work, $N_r(T)$ is determined by $S_r(T)$ up to the boundary correction in \eqref{eq:Nr_from_Sr_freerunning}, so $S_r(T)$ alone would have sufficed asymptotically.

\subsection{Remarks on our modeling assumptions}
\label{subsubsection:implicit_modeling_assumptions}
We make explicit a few of the assumptions underlying the model formulated above and discuss the resulting limitations. 
First, we assume that the dead time \(D\) is known and not part of the parameter \(\theta\) to be estimated. In practical settings where it is unknown, there is an extensive literature on dead-time estimation developed in nuclear physics, optics, and neuroscience (where the analogous quantity appears as the refractory period in spike trains)~\cite{knoll2010radiation, meeks2008dead, usman2018radiation, hampel2008estimation, akyurek2021new}. 

Second, our exposition adopts the \emph{nonparalyzable} (Type-I) dead-time model, which  accurately describes a commonly used class of single-photon detectors, modern single-photon avalanche diodes (SPAD)~\cite{qian2023modeling}. A \emph{paralyzable} (Type-II) variant, in which arrivals during an active
dead-time interval extend that interval, also appears in the literature and is
particularly important for modeling dead time in Geiger--M\"uller
tubes~\cite{knoll2010radiation,muller1973dead}.   We leave the analysis of the Type-II model to future work. Our model also accounts only for detector-induced dead time $D$. In TCSPC systems, an additional timing-electronics dead time can prevent the electronics from registering a new event even after the detector has recovered~\cite{rapp2021high}. This effect is often negligible, for example when the electronics dead time is shorter than the detector dead time~\cite{rapp2021high,qian2023modeling,kitichotkul2025free}.

Finally, we briefly comment on our choice of a discrete-time model.  While all digitally acquired data are ultimately discrete, the discretization scale can in principle be made arbitrarily fine. In that regime, the large-sample asymptotics developed in the remainder of this paper may fail to describe the effective statistical behavior, and one could instead consider a joint limit $T,K\to\infty$ or a continuous-time formulation. We deliberately fix $K$ and let only $T$ grow, which is consistent with the convention in simulations in the existing literature, e.g.~for single-photon detection and their associated inference problems~\cite{rapp2019dead, gupta2019photon, po2022adaptive, isbaner2016dead}.   More fundamentally, this has a practical systems-level motivation. In high-flux scenarios with large numbers of events per second such as  array-based SPAD architectures~\cite{zhang201830, hutchings2019reconfigurable, zhang2021240, kumagai20217, kitichotkul2025free}, storing and transmitting high-resolution timestamps imposes severe memory and bandwidth constraints. Many modern systems therefore perform relatively \emph{coarse} histogramming (in the sense $K \ll T$) directly at the sensing stage prior to data transmission.

\section{Asymptotic estimation bounds and Fisher information}
\label{section:asymptotic_theory}

We develop a statistical theory of estimation from DED observations $\{(G_t, Y_t)\}_{t=0}^{T-1}$ as defined in  \Cref{def:DED} in the limit $T \rightarrow \infty$. This asymptotic  regime affords tractable analysis and reflects the practical operating conditions in which the available signal is large relative to the number of parameters~\cite{kirmani2014first, becker2012fluorescence, tuo2024revisiting} (see also our previous discussion in \Cref{subsubsection:implicit_modeling_assumptions}). In \Cref{section:lan} we verify that the DED model is locally asymptotically normal (LAN)~\cite{le2000asymptotics, van2000asymptotic, bickel1993efficient} under mild regularity on the rate function, and identify the corresponding Fisher information rate $\mathcal I(\theta;\gamma)$. In \Cref{section:lower_bounds} we use this expansion to derive sharp lower bounds on the asymptotic distribution of estimation errors $\sqrt{T}(\hat\theta_T - \theta)$ via the H\'ajek convolution and local asymptotic minimax theorems.

\subsection{Regularity assumptions and local asymptotic normality of DED processes}
\label{section:lan}
Consider a DED process $\{(G_t, Y_t)\}_{t\geq 0}$ with $K$-periodic event probabilities $p_{t}(\theta)$ and gating-frequency vector $\gamma(\theta)\in[0,1]^K$. With $X_t$ and $P_{\theta,T}$ as in \Cref{subsection:likelihood_suff}, we study estimation in the limit $T \to \infty$. As in \eqref{eq:reparameterization}, we define the rate function $\lambda_t(\theta) \in [0, +\infty]$ by
$$
p_t(\theta) = 1 - \exp(-\lambda_t(\theta)),
$$
with the convention $\exp(-\infty) = 0$. This reparameterization is fully general and is adopted purely for convenience as in physical applications, $p_t(\theta)$ typically arises from a counting model in which $\lambda_t(\theta)$ is the more natural quantity to work with.

By $K$-periodicity of $\lambda_t(\theta)$, we write \(\lambda_r(\theta)\), \(r=0,\dots,K-1\), for the phasewise rates, so that
$$
\lambda_t(\theta) = \lambda_{t\bmod K}(\theta) \, .
$$
The phasewise rate vector \((\lambda_r(\theta))_{r=0}^{K-1}\) encodes all of the application-specific physical modeling and will be the central object on which our regularity assumptions are placed.  We make the following assumption on our DED process, which will be used to establish LAN, estimation lower bounds, and the efficiency of several estimators at our true data-generating parameter $\theta_0\in\Theta$.

\begin{assumption}
    \label{ass:lan}
    Assume that \(\theta_0\in\operatorname{int}(\Theta)\) and that all of the following hold.
\begin{enumerate}
    \item There exist constants \(0<\lambda_{\min}\le \lambda_{\max}<\infty\) such that \(\lambda_{\min}\leq \lambda_r(\theta) \leq \lambda_{\max}\) for all \(r=0,\dots,K-1\) and all \(\theta \in \Theta\).
    \item For each $r=0,\dots,K-1$, the map $\lambda_r:\Theta\to\mathbb R$ extends to a $C^3$ map on an open set containing \(\Theta\).
    \item If $\lambda_r(\theta) = \lambda_r(\theta_0)$ for all $r$ such that $\gamma_r(\theta_0)>0$, then $\theta = \theta_0$.
\item $\Theta$ is compact. 
    \item Define the phase-wise Fisher information matrices
\begin{equation}\label{eq:Ir_def}
\mathcal I_r(\theta)
\coloneq \frac{1-p_r(\theta)}{p_r(\theta)}\nabla \lambda_r(\theta)\nabla \lambda_r(\theta)^\top,
\end{equation}
For any $\alpha\in[0,1]^K$, define the information \emph{rate}
\begin{equation}\label{eq:Irate_def}
\mathcal I(\theta;\alpha)
\coloneq \frac{1}{K}\sum_{r=0}^{K-1}\alpha_r\mathcal I_r(\theta).
\end{equation}
We write the information matrix at the true parameter as
\begin{equation}\label{eq:I0_def}
\mathcal I_0 \coloneq  \mathcal I(\theta_0;\gamma(\theta_0)).
\end{equation}
Assume that the matrix $\mathcal I_0$ in \eqref{eq:I0_def} is positive definite.
\end{enumerate}
\end{assumption}
Importantly, the form of the Fisher information \(\mathcal I_0 \) first derived here separates the statistical information in the Bernoulli experiments from the effect of the acquisition rule in \eqref{eq:causality}. The phasewise matrices \(\mathcal I_r(\theta)\) quantify the information available when phase \(r\) is sampled, while the decision rule enters only through the limiting sampling frequencies \(\gamma_r\). Consequently, for the purpose of asymptotic information calculation, all that matters about the gating rule is how often each phase bin is observed. In many acquisition schemes, these frequencies can be obtained from the stationary law of the induced finite-state process or estimated by Monte-Carlo simulation \cite{rapp2019dead}. 

Conditions (1), (3), and (5) all serve to ensure that the estimation problem is well-posed. Specifically, Condition (1) rules out non-informative phases in which the detection probability is identically $0$ or $1$, and gives the global boundedness of the log-likelihood terms used in the MLE consistency proof. Both extremes are physically excluded because an arbitrarily small but nonzero rate is guaranteed by background radiation and intrinsic detector noise (e.g.~dark counts in SPADs), and an infinite rate is unphysical~\cite{taguchi2011modeling, qian2023modeling, buller2009single, highfluxspl_repo}. Conditions (3) and (5) are identifiability requirements at two different scales. Condition (3) is a global identifiability statement on the phases that the gating scheme actually visits ($\gamma_r > 0$): if two parameter values produce the same rate on every asymptotically visited phase, the $T$-scale limiting likelihood cannot distinguish them. Condition (5) ensures that the visited phases carry information along every parameter direction near $\theta_0$. 

Conditions (2) and (4) are technical. The $C^3$ smoothness in (2) is used to control the remainder in the LAN expansion and could be replaced by weaker differentiability notions (e.g.~differentiability in quadratic mean~\cite{van2000asymptotic}) at the cost of more involved proofs, but we have not pursued this because every rate model we are aware of in the applications of interest is smooth. Compactness of $\Theta$ in (4) is invoked in the consistency arguments of \Cref{section:efficiency} and is similarly stronger than strictly necessary. In practice, however,  parameters of interest such as depths, lifetimes, periods, and amplitudes are always known to lie in bounded physical ranges, and we are not aware of a DED application in which estimation over an unbounded parameter space would be meaningful.

Relying on \Cref{ass:lan}, our strategy is to establish local asymptotic normality (LAN) of DED processes.  LAN refers to statistical models whose log-likelihood, rescaled around the true parameter $\theta_0$ at scale $T^{-1/2}$, converges to a quadratic form so that it locally behaves like a Gaussian shift experiment. The curvature of this quadratic form is the Fisher information $\mathcal I_0$ and approximate Gaussianity yields sharp asymptotic lower bounds on estimation error~\cite{le2000asymptotics, van2000asymptotic}.
\begin{proposition}[LAN]
    \label{prop:lan} Suppose \Cref{ass:lan} holds at parameter $\theta_0\in \Theta$.  Define the normalized score
\begin{equation}
\label{eq:score_def_corrected}
\Delta_t(\theta_0)
\coloneq 
\frac1{\sqrt t}\sum_{r=0}^{t-1}
G_r \frac{Y_r-p_r(\theta_0)}{p_r(\theta_0)} \nabla \lambda_r(\theta_0).
\end{equation}
Then:
\begin{enumerate}
    \item For every bounded sequence $h_t$, the local alternatives $P_{\theta_0+h_t/\sqrt t,t}$ and $P_{\theta_0,t}$ are contiguous with respect to each other.
\item The following local asymptotic normality (LAN) holds:
 \begin{subequations}
    \label{eq:lan_expansion}
    \begin{align}
\Lambda_t(&\theta_0 + h_t/\sqrt{t}, \theta_0) \coloneq      \log\frac{dP_{\theta_0+h_t/\sqrt t,t}}{dP_{\theta_0,t}}\\
   & =
    h_t^\top \Delta_t(\theta_0)
    -\frac12 h_t^\top \mathcal I(\theta_0;\gamma) h_t
    +o_{P_{\theta_0,t}}(1).
    \end{align}
    \end{subequations}
    \item     Under $P_{\theta_0,t}$, $\Delta_t(\theta_0)\Rightarrow \Normal(0,\mathcal I_0)$.
\end{enumerate}
\end{proposition}
Note that for bounded sequences $h_t$, we have $\theta_0+h_t/\sqrt t\in \Theta$ for all sufficiently large $t$, such that the statements above are well-defined.    See, e.g.,~Sec.~3 of \citeasnoun{le2000asymptotics} for a formal definition of contiguity and the likelihood ratio employed above. The proof of \Cref{prop:lan} is conceptually similar to the i.i.d.\ case without dead time~\cite[Theorem~7.2]{van2000asymptotic} where $\gamma_r \equiv 1$ and is presented in \Cref{section:proofs}.

\subsection{Lower bounds on estimation}
\label{section:lower_bounds}
Using standard LAN theory and the expansion in \Cref{prop:lan}, we can translate the Fisher information rate
\(\mathcal I(\theta;\gamma)\) into sharp lower bounds for estimation in dead-time--constrained experiments. These lower bounds are stated through the H\'ajek convolution theorem for regular estimators in \Cref{thm:hajek_convolution} and the local asymptotic minimax bound for arbitrary estimators in \Cref{thm:local_asymptotic_minimax}.

\paragraph{H\'ajek convolution theorem}
We first state a H\'ajek convolution theorem for DED processes. This theorem applies to \emph{regular} estimators, meaning estimators whose local asymptotic distribution is stable under perturbations of the true parameter at the \(T^{-1/2}\) scale. As a reminder, we give the formal definition following \citeasnoun{le2000asymptotics}, Section~6, Theorem~3. 
\begin{definition}
\label{def:regular_estimator}
A sequence of estimators $(\hat\theta_T)_{T\ge0}$ is called regular at $\theta_0$ with limit law $\mathcal L_{\theta_0}$ on $\R^d$ if, for every fixed sufficiently small $h\in\R^d$,
$$
\sqrt{T}\Bigl(\hat\theta_T - \theta_0 - \frac{h}{\sqrt{T}}\Bigr)
 \Rightarrow
\mathcal L_{\theta_0}
$$
under $\theta_0 + h/\sqrt{T}$ as $T\rightarrow \infty$.
\end{definition}
The regularity condition excludes pathological superefficient behavior at a Lebesgue-null set of parameters $\theta\in\Theta$. For example,  a classical theorem of Le Cam implies that for any fixed sequence of  estimators $\hat \theta_T$ such that $\sqrt{T}(\hat\theta_T-\theta)$ has a limit law for each $\theta \in \Theta$, the set of parameter values at which regularity fails is a Lebesgue-null set~\cite[Chapter~6, Proposition~7]{le2000asymptotics}.  Thus lower bounds proved for regular estimators apply to arbitrary convergent estimators at Lebesgue-almost every parameter value.
\begin{theorem}[H\'ajek convolution theorem]
\label{thm:hajek_convolution}
Assume the hypotheses of \Cref{prop:lan} hold, and let $(\hat\theta_T)_{T\ge0}$ be a regular estimator at $\theta_0$ with limit law $\mathcal L_{\theta_0}$ as in \Cref{def:regular_estimator}.  Then there exists a probability measure $\nu_{\theta_0}$ on $\R^d$ such that
$$
\mathcal L_{\theta_0}
=
\Normal\left(0,\mathcal I_0^{-1}\right) * \nu_{\theta_0} \, .
$$
\end{theorem}
A short proof by applying LAN and the result of \citeasnoun{le2000asymptotics}, Section~6, Theorem~3, is given in \Cref{section:proofs}. 
The theorem says that no regular estimator can have smaller asymptotic spread than \(\Normal(0,\mathcal I_0^{-1})\).
The following risk bound is an immediate consequence that follows from weak convergence, Portmanteau's Theorem, and the monotone convergence theorem.
\begin{corollary}
\label{cor:hajek_lower_bound}
Assume the hypotheses of \Cref{prop:lan} hold. Let $(\hat\theta_T)$ be a regular estimator at $\theta_0$. Then,
\begin{equation}
\label{eq:hajek_risk_lower_bound}
\lim\limits_{b\rightarrow \infty}\liminf_{T\to\infty}
\E_{\theta_0}\left[b\wedge T\|\hat\theta_T - \theta_0\|_2^2\right]
 \ge
\Tr \bigl(\mathcal I_0^{-1}\bigr)
\end{equation}
where $x\wedge y\coloneq \min\{x,y\}$. This lower bound is achieved if the limit law in  \Cref{def:regular_estimator} is $\Normal\left(0,\mathcal I_0^{-1}\right)$.
\end{corollary}

\paragraph{Local minimax bound}
We next give a complementary characterization of the minimal asymptotic estimation error through a local asymptotic minimax result.  Unlike the H\'ajek convolution theorem above, which is stated for regular estimators, the local asymptotic minimax theorem applies to \emph{arbitrary} estimator sequences.  It shows that no estimator can uniformly improve on the same Gaussian Fisher information lower bound  over shrinking neighborhoods of the true parameter at scale  \(T^{-1/2}\).
The theorem is stated for bowl-shaped loss functions~\cite{le2000asymptotics}:
\begin{definition}
A function $W:\R^d\to[0,\infty)$ is called \emph{bowl-shaped}  if for every $a\ge 0$, the sublevel set
$$
\{u\in\R^d: W(u)\le a\}
$$
is closed, convex, and symmetric about the origin.
\end{definition}
Examples include \(W(u)=\|u\|_p^q\) for \(p\ge 1\) and \(q>0\), in particular the squared Euclidean loss \(W(u)=\|u\|_2^2\).
\begin{theorem} \label{thm:local_asymptotic_minimax}
Assume the hypotheses of \Cref{prop:lan} hold.  Let $W:\R^d\to[0,\infty)$ be bowl-shaped. 
Then for every sequence of estimators $(\hat\theta_T)$, the following local minimax bound holds
\begin{subequations}
\label{eq:laminimax_general}
\begin{align}
\lim_{b\to\infty} \lim_{c\to\infty} \liminf_{T\to\infty}
&\sup_{\|\theta-\theta_0\|_2\le \frac{c}{\sqrt{T}}}
\E_{\theta}
 \left[
b\wedge
W \left(
\sqrt T \Bigl(\hat\theta_T-\theta\Bigr)
\right)
\right]
\\
&\ge
\E\left[ W \left( Z\right)\right],
\end{align}
\end{subequations}
where $Z\sim \Normal(0,\mathcal{I}_0^{-1} )$.  In particular,
\begin{align*}
\lim_{b\to\infty}\ \lim_{c\to\infty}\ \liminf_{T\to\infty}\
&\sup_{\|\theta-\theta_0\|_2\le c/\sqrt{T}}
\E_{\theta}
 \left[
b\wedge
\left\|
\sqrt T\Bigl(\hat\theta_T-\theta\Bigr)
\right\|^2_2
\right] \\
&\ge \Tr(\mathcal{I}_0^{-1}).
\end{align*}
\end{theorem}
The proof can be found in \Cref{section:proofs}.  The truncation parameter \(b\) in \Cref{cor:hajek_lower_bound} and \Cref{thm:local_asymptotic_minimax} can be removed if one cares only about the lower bound and not about its attainment. The Portmanteau Theorem implies that, under the hypotheses of \Cref{prop:lan}, every sequence of estimators \((\hat\theta_T)_{T\geq 0}\) satisfies
$$
\lim_{c\to\infty}\ \liminf_{T\to\infty}\
\sup_{\|\theta-\theta_0\|_2 \le c/\sqrt{T}}
\E_{\theta}\!\left\|
\sqrt{T}\bigl(\hat\theta_T - \theta\bigr)
\right\|_2^2
\;\ge\; \Tr(\mathcal{I}_0^{-1}) \, ,
$$
together with the analogous statement for \Cref{cor:hajek_lower_bound}. As is standard in asymptotic statistics, we retain the truncation in the stated versions because it ensures the bound is \emph{attained} by any estimator with limit law \(\sqrt{T}(\hat\theta_T - \theta_0) \Rightarrow \Normal(0, \mathcal I_0^{-1})\). Without truncation, attainment would additionally require passing from convergence in distribution to convergence of second moments, for example via the Vitali convergence theorem  by showing uniform integrability of \(\left(\|\sqrt{T}(\hat\theta_T - \theta_0)\|_2^2\right)_{T\geq0}\).

This justifies restricting the error analysis of estimators $(\hat\theta_T )_{T\geq 0}$  in the rest of the paper to the limit law of the error \(\sqrt{T}(\hat\theta_T - \theta_0)\). An estimator whose error distribution converges to the minimum-spread law \(\Normal(0,\mathcal I_0^{-1})\) is asymptotically optimal in the sense of \Cref{thm:hajek_convolution} and \Cref{thm:local_asymptotic_minimax} and is called \emph{asymptotically efficient}, or simply \emph{efficient}, in statistics~\cite{van2000asymptotic,le2000asymptotics}. In the next section, we give two practical examples of efficient estimators for DED processes.

\subsection{On the assumption of gating-frequency convergence}
\label{subsubsection:circumventing_gating_frequency_convergence}
We briefly comment on the role of gating-frequency convergence in \Cref{def:DED}. The assumption ensures that the acquisition rule has a well-defined asymptotic information rate, which is the quantity entering the lower bounds in \Cref{thm:hajek_convolution} and \Cref{thm:local_asymptotic_minimax}. One \textit{could} prove these results without this assumption by passing to subsequences along which the empirical gating frequencies converge to some limiting vector \(\gamma\) achieving the liminf risk, and apply the same LAN and lower-bound arguments along this subsequence. We nonetheless impose gating-frequency convergence directly, for three reasons. First, there is no additional statistical insight when not imposing this assumption. Second, it enables simulation of \(\gamma\) and the empirical computation of asymptotic quantities (as for example in \Cref{subsection:one_step}), which would otherwise be practically infeasible. Finally, it is the practically relevant case as practically relevant schemes such as the free-running and synchronous schemes satisfy it (\Cref{app:gating_frequency_convergence_acquisition}).  This is natural from an engineering standpoint, since a rule that has not settled asymptotically has no stable answer to the question of which bins are most informative to sample. If the empirical gating frequencies oscillate indefinitely between two limits, for instance, one should restrict to the gating-frequency convergent subsequence achieving the lower asymptotic risk.
\section{Efficient estimators}
\label{section:efficiency}

The lower bounds in \Cref{cor:hajek_lower_bound} and \Cref{thm:local_asymptotic_minimax} characterize the best asymptotic performance any estimator can hope to achieve by a limiting mean-zero normal law with covariance $\mathcal I_0^{-1}$. In this section, we will present two concrete implementable estimators that attain these bounds and are hence asymptotically efficient.

In \Cref{subsection:mle} we show asymptotic efficiency of  the maximum likelihood estimator (MLE), variants of which are widely used in practical DED applications~\cite{coates1968correction,  pediredla2018signal, rapp2019dead, kitichotkul2025free, maus2001experimental, chessel2013maximum, bajzer1991maximum, wang2026synchronous}. In \Cref{subsection:one_step} we introduce one-step estimators for DED processes as a more flexible alternative. Originally popularized by Le Cam, a one-step estimator applies a single local Newton correction to a sufficiently good initial pilot estimator,  attaining the same asymptotic efficiency as the MLE while avoiding iterative optimization.

\subsection{Maximum likelihood estimation}
\label{subsection:mle}
Given data $X_0, X_1, \ldots, X_{T-1}$, a maximum likelihood estimator (MLE) is defined as a sequence of estimators $\hat \theta_T\in \Theta$ satisfying
\begin{equation}
    \label{eq:mle}
\begin{aligned}
\hat\theta_T
&\in \arg\max_{\theta\in\Theta}\ell_T(\theta),\\
\ell_T(\theta)
&\coloneq  \log P_{\theta,T}(X_0, X_1, \ldots, X_{T-1})
\end{aligned}
\end{equation}
for each $T$, where the explicit form of the likelihood is given in \Cref{eq:phase_loglikelihood_intro}. We refer to any such sequence as ``the'' MLE, although it need not be unique. 
The MLE is a well-known inference procedure that is widely used in practical applications involving DED processes, though prior work has only addressed specific gating rules~\cite{coates1968correction,rapp2019dead,kitichotkul2025free,maus2001experimental,chessel2013maximum,bajzer1991maximum,wang2026synchronous}. In contrast, the likelihood formula in \Cref{eq:phase_loglikelihood_intro} enables computation of the MLE for any $\theta$-independent gating rule.

One of the reasons for its popularity is that it is often asymptotically efficient in the i.i.d.-case~\cite{van2000asymptotic}. The following theorem establishes  asymptotic efficiency of the MLE in the non-i.i.d.\ DED-process case under \Cref{ass:lan}.
\begin{theorem}
\label{thm:mle_efficient} Consider the hypotheses of \Cref{prop:lan}.  Write  \(\mathcal I_0\coloneq \mathcal I(\theta_0;\gamma(\theta_0))\). Then for any MLE $\hat\theta_T$ the following are true.
\begin{enumerate}
\item $\hat\theta_T$ is consistent, meaning that $\hat\theta_T\xrightarrow{P_{\theta_0,T}}\theta_0$.
\item It is efficient and, in particular, asymptotically normal $\sqrt{T} (\hat\theta_T-\theta_0) \Rightarrow
\Normal(0,\mathcal I_0^{-1})$.
\item $\hat\theta_T$ is regular meaning that for every fixed $h\in\R^d$,
$$
\sqrt{T}\left(\hat\theta_T-\theta_0-\frac{h}{\sqrt{T}}\right)
\Rightarrow
\Normal(0,\mathcal I_0^{-1})
$$
under $P_{\theta_0+h/\sqrt{T},T}$.
\end{enumerate}
\end{theorem}
In particular, the MLE achieves the lower bound \eqref{eq:hajek_risk_lower_bound} in the Hájek convolution theorem among regular estimators, and the lower minimax bound \eqref{eq:laminimax_general} among all estimators.

A proof of \Cref{thm:mle_efficient} can be found in \Cref{section:proofs}.  It is worth understanding the intuition behind this result, as it will inform the following discussion. The key is the typical local quadratic structure of the log-likelihood due to LAN (\Cref{prop:lan}), meaning that in a $T^{-1/2}$-neighborhood of $\theta_0$ we have informally that
\begin{equation}
\label{eq:likelihood_equation}
\ell_T(\theta) \approx \ell_T(\theta_0) + (\theta-\theta_0)^\top \nabla\ell_T(\theta_0) - \tfrac{T}{2}(\theta-\theta_0)^\top \mathcal I_0 (\theta-\theta_0),
\end{equation}
so that maximizing the log-likelihood is asymptotically equivalent to maximizing this quadratic. Its maximizer is therefore, approximately,
\begin{equation}
\label{eq:max_likelihood}
\theta_0 + \mathcal I_0^{-1}\nabla\ell_T(\theta_0)/T,
\end{equation}
which, after rescaling, has the efficient limit law $\Normal(0, \mathcal I_0^{-1})$.

\subsection{One-step estimators}
\label{subsection:one_step}
The argument above shows that the MLE achieves efficiency by exploiting the local quadratic LAN-structure of the log-likelihood around $\theta_0$. Once this structure is recognized, however, the full and typically nonconvex  optimization required to compute the MLE is no longer needed.

Replacing $\theta_0$ in \eqref{eq:likelihood_equation} by any sufficiently accurate \(\Theta\)-valued pilot estimator $\tilde\theta_T$ such that the quadratic approximation remains valid, the unique maximizer of the resulting quadratic in $\theta$ is the explicit Newton update
\begin{equation}
\label{eq:one_step_formula}
\hat\theta_T^{\mathrm{os}}
=
\tilde\theta_T
+
\bigl[\mathcal I(\tilde\theta_T;\hat\gamma_T)\bigr]^{-1}
\frac{U_T(\tilde\theta_T)}{T}
\end{equation}
if this lies in $\Theta$\footnote{Under the assumptions of \Cref{thm:onestep_efficient}, this constraint is satisfied with probability tending to one since \(\theta_0\in\operatorname{int}(\Theta)\) and \(\tilde\theta_T\xrightarrow{P}\theta_0\) in the setting of \Cref{thm:onestep_efficient}. } where $\hat\gamma_T$ is an estimator of the limiting gating frequencies $\gamma$ and $U_T(\tilde\theta_T)\coloneq \nabla \ell_T(\tilde\theta_T)$ is called the \textit{score} in statistics. Intuitively, this estimator should inherit the same efficiency as the MLE since it solves the same local quadratic problem. This general observation in the i.i.d.-setting is due to Le Cam, who introduced the  \emph{one-step estimator} as a tool for converting any $\sqrt T$-consistent pilot into an asymptotically efficient procedure~\cite{kraft1956remark, le1956asymptotic, le2000asymptotics, van2000asymptotic}. We apply this construction to the DED setting. In particular, we describe how to estimate $\gamma$ from the gating sequence, how to use this to compute the one-step update, and prove that the resulting estimator inherits the efficient $\Normal(0, \mathcal I_0^{-1})$ limit.

\paragraph{Computing the one-step update (OSE)}
In order to compute the one-step update in \eqref{eq:one_step_formula}, we need an estimate of the limiting gating frequencies \(\gamma(\theta_0)\) of our DED process.  These frequencies can be empirically estimated directly from the observed open-gate indicators.  For each phase \(r=0,\ldots,K-1\), define
\begin{equation}
\label{eq:empirical_gating_fractions}
\hat\gamma_{r,T}
=
\frac1L\sum_{\ell=0}^{L-1}G_{r+\ell K} \, ,
\end{equation}
where \(L = \lfloor T/K\rfloor\), and write \(\hat\gamma_T=(\hat\gamma_{0,T},\ldots,\hat\gamma_{K-1,T})\). By gating-frequency convergence, \(\hat\gamma_{r,T}\xrightarrow{P_{\theta_0,T}} \gamma_r(\theta_0)\). Thus the limiting information-rate matrix \(\mathcal I(\theta_0;\gamma(\theta_0))\) is naturally estimated by \(\mathcal I(\tilde\theta_T;\hat\gamma_T)\), as we show more formally in the proof later. The one-step estimator can then be computed through the single Newton step in \eqref{eq:one_step_formula}. Evaluating \(\mathcal I(\tilde\theta_T;\hat\gamma_T)\) uses the phasewise sums in \eqref{eq:Ir_def}--\eqref{eq:Irate_def} and costs \(O(Kd^2)\), while evaluating the score \(U_T(\tilde\theta_T)=\nabla\ell_T(\tilde\theta_T)\) by differentiating the phasewise likelihood \eqref{eq:phase_loglikelihood_intro} costs \(O(Kd)\). The final \(d\)-dimensional linear solve in \eqref{eq:one_step_formula} costs \(O(d^3)\), so the one-step correction has total cost \(O(Kd^2+d^3)\), or simply \(O(K)\) for fixed parameter dimension ($d=3$ for our lidar example, below).

\paragraph{Asymptotic efficiency}
The following theorem, proved in \Cref{section:proofs}, establishes efficiency of the estimator in \eqref{eq:one_step_formula}. 
\begin{theorem}
\label{thm:onestep_efficient}
Assume the hypotheses of \Cref{prop:lan} hold at $\theta_0$.
Let $\tilde\theta_T$ be any \(\Theta\)-valued preliminary estimator such that $\sqrt T (\tilde\theta_T-\theta_0)=O_{P_{\theta_0,T}}(1)$ and define $\hat\theta_T^{\mathrm{os}} $ as in \eqref{eq:one_step_formula} on the event that $\mathcal I \bigl(\tilde\theta_T;\hat\gamma_T\bigr)$ is invertible, and as $\tilde\theta_T$ otherwise.
Then we have asymptotic efficiency
$$
\sqrt T (\hat\theta_T^{\mathrm{os}}-\theta_0)
\Rightarrow \Normal(0,\mathcal I_0^{-1}) \, .
$$
Moreover, for every fixed $h\in\R^d$,
$$
\sqrt T\left(\hat\theta_T^{\mathrm{os}}-\theta_0-\frac{h}{\sqrt T}\right)
\Rightarrow \Normal(0,\mathcal I_0^{-1})
$$
under $P_{\theta_0+h/\sqrt T,T}$. Therefore $(\hat\theta_T^{\mathrm{os}})$ is regular at $\theta_0$ with efficient limit law.
\end{theorem}
This efficiency result shows that the one-step construction has the same asymptotically efficient estimation rate as the MLE. However, there are two other important reasons that make the one-step estimator attractive compared to the MLE.   First, computing a one-step update from a closed-form pilot can be cheaper than running a full nonconvex MLE optimization (which in practice often has to be done approximately for speedup~\cite{kitichotkul2025free}). Second, it gives us the freedom to choose the pilot for stability rather than for accuracy, and have the one-step update recover efficiency. This decoupling is particularly valuable because the MLE itself is known to be sensitive to model misspecification~\cite{kraft1956remark, le1956asymptotic, le1990maximum, le2000asymptotics}. We will illustrate these advantages empirically in the next section for a DED process arising from lidar, comparing the one-step estimator and MLE more closely.

\section{Application to lidar}
\label{section:lidar}
We illustrate the theory developed in the previous sections on the important application of single-photon lidar and demonstrate its validity empirically. In lidar, a scene is illuminated by a periodic train of short laser pulses, and the returning photons are detected by a single-photon detector subject to dead time. The arrival-time statistics of these photons encode the round-trip time-of-flight to the illuminated scene point and the local reflectivity. Importantly, the time-of-flight  can be converted to a range  if the speed of light $c$ is known in the optical medium.
Inferring the time-of-flight, local reflectivity, and sometimes background radiation from photon detection times is the central task in single-photon lidar imaging~\cite{kirmani2014first, buller2009single, rapp2019dead, kitichotkul2025free}, and the resulting estimates form the basis of three-dimensional reconstruction as used in downstream sensing applications.

\subsection{Statistical model}
Lidar is modeled by a DED process with the rate function $\lambda_r(\theta)$ typically taking the form
\begin{equation}
\label{eq:lidar_lambda_model}
\lambda_r(\theta)  =  a f_{\tau}(r) + b,
\qquad r = 0, \dots, K-1. 
\end{equation}
The parameter to be inferred is $\theta = (a, \tau, b) \in [0,\infty)\times [0,K)\times [0,\infty)$  and the pulse shape is modeled as
\begin{equation}
\label{eq:binned_template}
f_{\tau}(r)\coloneq \int_{\tau-r}^{\tau-r+1} f(x) dx,
\end{equation}
where $f:\mathbb{R}\to[0,\infty)$ is a \(K\)-periodic time-continuous pulse shape with normalization  $\int_0^K f(x) dx=1$. Here \(K\in\N\) is the number of discrete phase bins per laser period, whose physical duration is denoted by \(t_r\). The delay $\tau\in[0,K)$ is the round-trip time-of-flight in phase-bin units, corresponding to physical time \(\tau t_r/K\). $a\ge 0$ is the reflectivity, and $b\ge 0$ is a constant background rate due to ambient light and detector dark counts. In practice, the pulse template  $f$ can either be determined by measuring the pulse shape of the laser empirically in a calibration step or through approximation with a parametric model~\cite{gupta2019asynchronous, gupta2019photon,kitichotkul2025free,po2022adaptive, kitichotkul2023role, shin2015photon, heide2018sub,wang2026synchronous,mccarthy2013kilometer}.

Importantly, although the underlying physical process is continuous in time, the discrete index $r$ in \eqref{eq:lidar_lambda_model} reflects the histogramming that is typically performed in practice with single-photon detectors for bandwidth reasons~\cite{zhang201830, hutchings2019reconfigurable, zhang2021240, kumagai20217, kitichotkul2025free,wang2026synchronous}, see also \Cref{subsubsection:implicit_modeling_assumptions}.  The statistical problem of estimating $\theta$ from DED process data $\{(G_t, Y_t)\}_{t = 0}^{T-1}$, therefore, naturally fits into our theory established in  \Cref{section:asymptotic_theory} and \Cref{section:efficiency}.

\subsection{Dead-Time-Free and Unbiased Bounds for the Lidar Estimation Problem}
We briefly summarize how existing classical theory that neglects dead-time effects~\cite{berk1972consistency, efron2022exponential, van2000asymptotic} has been applied to lower bounds for estimation in the lidar literature. Most of these results make the approximation that dead-time effects are negligible. Concretely, such bounds have been derived for unbiased estimators of range in a continuous-time model~\cite{kitichotkul2023role} and for unbiased estimators of a constant photon rate in discrete-time models~\cite{pediredla2018signal, gupta2019asynchronous, lu2013adaptive}. Besides not accounting for dead time, these results do not characterize fundamental lower bounds for possibly biased estimators as would be typical in practice. 

To our knowledge, only two prior works address estimation error under dead time: Daniel and Fessler~\cite{daniel2000mean}, who derive the asymptotic mean and variance of the detected photon count under a constant photon rate and free-running acquisition, and Wu et al.~\cite{wu2025performance}, who derive the same Fisher information as ours in the lidar-setting but only for the free-running gating scheme and only as a lower bound for unbiased estimators.

\subsection{Theory for asymptotic estimation in lidar}
We establish a fundamental estimation lower bound for the lidar problem \eqref{eq:lidar_lambda_model} that accounts for dead time and handles non-constant photon rates. We then verify the conditions of \Cref{ass:lan}, allowing us to apply the theory from  \Cref{section:asymptotic_theory,section:efficiency}.   We make the following assumptions in addition to the model definition in \eqref{eq:binned_template} on the true underlying data-generating parameter $\theta_0$ in our experiment.
\begin{assumption}
    \label{ass:lidar}
    \begin{enumerate}
        \item The pulse \(f:\R\rightarrow [0, \infty)\) is \(C^3\).
        \item Our parameter candidate set $\Theta\subset\R^3$ satisfies \(\theta_0\in\operatorname{int}(\Theta)\) and takes the form
$$
\Theta
=[a_-,a_+]\times [\tau_-,\tau_+]\times [b_-,b_+] \, ,
$$
for $0<a_-<a_+<\infty$, $0\leq\tau_-<\tau_+<K$, and $0<b_-<b_+<\infty.$
\item The single-photon gating rule $\Phi_t$ of the DED process as in \eqref{eq:causality} is causal, dead-time--constrained, and gating-frequency convergent with limiting gating frequencies \(\gamma(\theta_0) \in [0,1]^K\).
\item  For every $\theta \in \Theta$, the active phases identify the true lidar parameter as stated in \Cref{ass:lan} (3).
\item Writing $\theta_0=(a_0,\tau_0,b_0)$ and defining the gradient of $\lambda_r(\theta)$ as $v_r(\theta_0)\coloneq  \begin{pmatrix}
f_{\tau_0}(r)\\
a_0\bigl(f(\tau_0-r+1)-f(\tau_0-r)\bigr)\\
1
\end{pmatrix}$, the  set of vectors $\{v_r(\theta_0): \gamma_r(\theta_0)>0\}$ spans $\R^3$.
    \end{enumerate}
\end{assumption}
We briefly discuss this assumption. As pointed out in \Cref{section:asymptotic_theory}, the $C^3$ assumption is not strictly necessary, but is easily satisfied by practically used Gaussian pulse models. The upper bounds on \(a\) and \(b\) reflect finite signal and background levels, and the lower bounds \(a_->0\) and \(b_->0\) exclude degenerate regimes where $\tau$ is not identifiable and where there is no background radiation. The compactness assumption on \(\tau\) is purely technical. Since \(\tau\) is periodic, the natural parameter space is the torus \(\mathbb{R}/K\mathbb{Z}\), but we use a compact Euclidean interval to avoid formulating the LAN theory of \Cref{section:asymptotic_theory} on manifolds. The proofs use compactness only for uniform bounds on continuous functions, which hold equally on the torus, so all results apply with the torus \(\R/K\Z\)   in place of \([\tau_-,\tau_+]\). For the free-running and synchronous acquisition rules considered in \Cref{subsection:observation_model}, condition (3) follows from \Cref{app:gating_frequency_convergence_acquisition}. Identifiability in (4) and (5) is necessary for the statistical problem to be well-posed and must be checked for the pulse model used. One can easily verify the following proposition:
\begin{proposition}
A lidar model satisfying \Cref{ass:lidar} is a DED process satisfying \Cref{ass:lan}.
\end{proposition}
In particular, the correct notion of Fisher information is given by
$$
\mathcal I(\theta_0;\gamma)
=
\frac1{K}
\sum_{r=0}^{K-1}
\gamma_r
\frac{\exp(-\lambda_r(\theta_0))}
{1-\exp(-\lambda_r(\theta_0))}
 v_r(\theta_0)v_r(\theta_0)^\top
$$
evaluated at the correct gating frequencies, i.e.~$\mathcal I_0 =\mathcal I(\theta_0;\gamma(\theta_0)).$ We can directly apply \Cref{thm:hajek_convolution}, \Cref{thm:local_asymptotic_minimax}, \Cref{thm:mle_efficient}, and \Cref{thm:onestep_efficient} in order to conclude the following result.
\begin{corollary}
\label{cor:lidar_bounds_and_efficiency}
Suppose \Cref{ass:lidar} holds. Then the statements \Cref{thm:hajek_convolution}, \Cref{thm:local_asymptotic_minimax}, \Cref{thm:mle_efficient}, and \Cref{thm:onestep_efficient} hold for the lidar estimation problem, giving the efficient limiting distribution \(\mathcal N(0,\mathcal I_0^{-1})\) and the corresponding Fisher-information lower bounds for regular and local-minimax estimation.
\end{corollary}
Having proved asymptotic lower bounds and the efficiency of the MLE and one-step estimators, we now discuss practical estimators for lidar in greater detail, both to show how they fit into this picture and to empirically validate the asymptotics at practically meaningful sample sizes.

\subsection{Maximum likelihood and one-step estimators for lidar}
\label{subsection:lidar_estimators}
In practice, most of the lidar estimation literature is based, either explicitly or implicitly, on maximum likelihood estimation~\cite{rapp2017few, rapp2021high, shin2015photon, kitichotkul2023role,
kitichotkul2025free, gupta2019asynchronous, coates1968correction,
wang2026synchronous, liu2019pile}.   In our notation, the likelihood is the
phasewise likelihood in \Cref{eq:mle}, with
$$
p_r(\theta)=1-\exp(-\lambda_r(\theta)) \, ,
\qquad
\lambda_r(\theta)=a f_{\tau}(r)+b \, .
$$
The corresponding maximum likelihood estimator is
\begin{equation}
\label{eq:lidar_mle_definition}
\hat\theta_T^{\mathrm{MLE}}
\in
\arg\max_{\theta\in\Theta}
\ell_T(\theta).
\end{equation}
Computing this estimator exactly can be expensive, since the likelihood is generally nonconvex. While some practical methods propose direct computation~\cite{kirchhoff2025development, wang2026synchronous}, many practical methods in the literature replace the full joint optimization by faster approximations. One common strategy is to decouple the optimization over \(\tau\) and \((a,b)\), for example by using a bilevel optimization procedure~\cite{shin2015photon, kitichotkul2023role, kitichotkul2025free}.  Another strategy is to first estimate the underlying rate profile \(g\) by applying the closed-form maximum-likelihood estimator for the underlying intensity, commonly called the Coates-type correction, and then estimate the physical parameters \(\theta\) from this corrected profile rather than by directly maximizing the likelihood over \(\theta\) itself~\cite{coates1968correction, gupta2019asynchronous, liu2019pile, wang2026synchronous}.

In the experiments below, ``MLE initialized from a pilot'' refers to the output of a local numerical likelihood optimizer initialized at that pilot, not necessarily to the exact global maximizer in \eqref{eq:mle}. Strictly speaking,  the theoretical efficiency result only applies to the latter.

Apart from likelihood-based and likelihood-approximation methods, the only alternative we are aware of that is effective in settings where dead time plays a significant role is based on the method of generalized estimating equations~\cite{rapp2019dead}. As this approach has been reported to be more computationally demanding than maximum-likelihood-based  methods without yielding a clear performance advantage
\cite{kitichotkul2025free}, we focus here only on maximum likelihood and on new estimators motivated by the one-step construction.

A key observation in \Cref{section:efficiency} was that the asymptotic efficiency of the MLE comes from the local quadratic structure of the log-likelihood, not from global likelihood maximization itself. The one-step estimator exploits this structure directly by starting from a pilot estimator that is close enough to the true parameter such that a single Newton correction step produces an estimator with the same asymptotic efficiency as the MLE.

Based on this, our theory suggests a different estimator-design principle for the lidar example. Rather than constructing increasingly accurate approximations to the global MLE through iterative optimization, one can construct simple pilot estimators whose only task is to reach the correct local region of the parameter space, and then convert them into efficient estimators by a one-step update. This separates the global localization problem from the local efficiency problem and allows the pilot estimator to be designed for speed, numerical stability, or robustness, rather than as a direct approximation of the MLE. We illustrate this principle by introducing two  new pilot estimators.

Both pilot estimators first convert the observed phase counts \((N_r(T),S_r(T))\) as defined in \Cref{section:model_formulation} into estimated photon rates through
\begin{equation}
\label{eq:pr_est}
\hat p_r
=
\left(
\frac{S_r(T)+1/2}{N_r(T)+1}
\right), \qquad \hat\lambda_r
=
-\log(1-\hat p_r).
\end{equation}
The \(1/2\)- and \(1\)-offsets regularize the empirical detection frequency and avoid the boundary cases \(S_r(T)=0\) and \(S_r(T)=N_r(T)\).\footnote{Equivalently, this is the posterior mean under the Jeffreys prior \(p_r\sim\BetaDist(1/2,1/2)\) for the binomial detection probability.}
\paragraph{Fourier inversion estimator}
As a first example, we use a closed-form method-of-moments pilot based on the zeroth and first Fourier modes of the log-transformed rate profile. Let
\[
\hat M_k
=
\frac1K
\sum_{r=0}^{K-1}
\hat\lambda_r \exp(2\pi i k r/K),
\qquad k=0,1.
\]
With \(d_1\) denoting the first binned Fourier coefficient of the pulse template, as derived in \Cref{app:lidar_pilot_derivations}, the Fourier inversion (FI) estimator is \(\hat\theta_T^{\mathrm{FI}}\coloneq  (\hat a_T^{\mathrm{FI}}, \hat\tau_T^{\mathrm{FI}},\hat b_T^{\mathrm{FI}})\), where
\begin{subequations}
\label{eq:lidar_fourier_estimator}
\begin{align*}
\hat a_T^{\mathrm{FI}}
&=
\frac{|\hat M_1|}{|d_1|},
\qquad
\hat\tau_T^{\mathrm{FI}}
=
\frac{K}{2\pi}
\arg\left(\frac{\hat M_1}{d_1}\right),\\
\hat b_T^{\mathrm{FI}}
&=
\hat M_0-\frac{\hat a_T^{\mathrm{FI}}}{K}.
\end{align*}
\end{subequations}
Here \(\arg\) denotes the argument with values in \([0,2\pi)\), so that \(\hat\tau_T^{\mathrm{FI}}\in[0,K)\).
This estimator is  fully explicit and requires no optimization, but it relies on a single harmonic and is therefore sensitive to localized fluctuations or model deviations.

\paragraph{Robust matched estimator}
The robust matched estimator uses the same rate estimates but replaces single-harmonic inversion by median-centered template matching. First estimate the background by the phasewise median
\[
\hat b_T^{(0)}
=
\operatorname{median}_{r}\hat\lambda_r,
\qquad
J_r=\hat\lambda_r-\hat b_T^{(0)}.
\]
Then estimate the delay by maximizing the circular matched-filter score
\begin{equation}
\label{eq:robust_range}
\hat\tau_T^{\mathrm{rob}} \in \arg\max_{j=0,\ldots,K-1} C_j \, , \; \text{where }
C_j=\sum_{r=0}^{K-1}J_r f_j(r).
\end{equation}
Given this delay estimate, project the centered profile onto the shifted template and re-estimate the background from the median residual:
\[
\begin{aligned}
\hat a_T^{\mathrm{rob}}
&=
\left(
\frac{
\sum_{r=0}^{K-1} J_r f_{\hat\tau_T^{\mathrm{rob}}}(r)
}{
\sum_{r=0}^{K-1} f_{\hat\tau_T^{\mathrm{rob}}}(r)^2
}
\right)_+,
\\
\hat b_T^{\mathrm{rob}}
&=
\operatorname{median}_{r}
\left(
\hat\lambda_r-\hat a_T^{\mathrm{rob}}
f_{\hat\tau_T^{\mathrm{rob}}}(r)
\right),
\end{aligned}
\]
where $(x)_+\coloneq \max(x,0).$ The full estimator is
\[
\hat\theta_T^{\mathrm{rob}}
=
\left(
\hat a_T^{\mathrm{rob}},
\hat\tau_T^{\mathrm{rob}},
\hat b_T^{\mathrm{rob}}
\right).
\]
Its main feature is that it uses the full pulse shape rather than only the first Fourier mode, while the two median steps make it less sensitive to localized finite-sample fluctuations and mild background mismatch. A short derivation from robust least squares and matched filtering is given in \Cref{app:lidar_pilot_derivations}. Similar estimators based on photon arrival quantiles, designed to perform reliably in small-sample regimes and suppress noise-induced detections, have also been proposed in~\citeasnoun{rapp2017few}. The correlation-based estimation of \(\tau\) used here is likewise used often in lidar and frequently incorporated into iterative (approximate) ML algorithms~\cite{shin2015photon, kitichotkul2023role, kitichotkul2025free}.  Moreover, the estimator described here can be computed in \(O(K\log K)\) time.  The estimates \(\hat p\), \(\hat\lambda\), \(\hat a_T^{\mathrm{rob}}\), and \(\hat b_T^{\mathrm{rob}}\) require only phasewise operations or sums, and hence are \(O(K)\), the median steps are also \(O(K)\)~\cite{cormen2022introduction}, and the  correlation \(C_j\)  can be computed with an FFT-based algorithm in \(O(K\log K)\) time~\cite{cormen2022introduction}.  As explained in~\Cref{subsection:one_step}, the subsequent one-step estimate is also $O(K)$.

\subsection{Simulation study}
\label{section:simulation}
We show how our asymptotic theory manifests in finite samples. In particular, we verify that the MLE and the one-step estimators approach the Fisher-information lower bound at practical horizon sizes, compare their numerical behavior at practically relevant finite horizons, and test their performance against model misspecification.

The pulse shape we will use is a wrapped Gaussian, letting $\phi_\sigma(z)
=
\frac{1}{\sqrt{2\pi}\sigma}
\exp \left(-\frac{z^2}{2\sigma^2}\right),$ and defining
\begin{equation}
\label{eq:wrapped_gaussian}
f(x)=\sum_{\ell\in\mathbb Z}\phi_\sigma(x+\ell K) \, .
\end{equation}
This wrapped-Gaussian pulse model is standard in single-photon lidar simulations~\cite{rapp2017few,kitichotkul2023role,kitichotkul2025free}. The binned template \(f_\tau(r)\) is then defined as in \eqref{eq:binned_template}.

We use a physical laser period of \(t_r=100 \mathrm{ns}\), divided into \(K=1000\) bins, giving bin width \(\Delta=t_r/K=0.1 \mathrm{ns}\). The pulse width is \(\sigma=10\) bins, i.e.~\(\sigma\Delta=1 \mathrm{ns}\). The true parameter to recover is $\theta_0=(a_0,\tau_0,b_0)=(1, 370.4, 0.003) ,$ so the expected signal return is \(a_0=1\) photon per period and the expected background is \(K b_0=3\) photons per period. The dead time is \(D=500\) bins, corresponding to a physical dead time of  \(t_d = D\Delta=50 \mathrm{ns}\). Under the free-running acquisition scheme, the detector advances by one bin after a miss and by \(D+1\) bins after a detection. These parameter values were chosen as a representative case of the simulation settings commonly used in the single-photon lidar literature~\cite{arvani2018direct, gupta2019asynchronous,rapp2019dead, incoronato2021statistical, kitichotkul2025free}.

We first validate the asymptotic efficiency of the MLE, in particular subject to different initializations. 
Since our theory characterizes the efficiency of the MLE, we test it by directly optimizing the joint likelihood \eqref{eq:lidar_mle_definition}, as in \citeasnoun{wang2026synchronous} and \citeasnoun{kirchhoff2025development}. The optimization is performed in latent variables \(\eta=(\eta_a,\eta_{\tau},\eta_b)\) using $a=\exp(\eta_a)$, $\tau=\frac{K}{1+\exp(-\eta_{\tau})}$, and  $b=\exp(\eta_b).$ This parametrization enforces positivity of \(a\) and \(b\), and keeps \(\tau\in(0,K)\). Since the true value \(\tau_0=370.4\) is far from the period boundary, the use of this Euclidean representative of the circular delay does not create a boundary issue in the present experiment. We use the NLopt implementation of L-BFGS~\cite{NLopt, LBFGS}, with maximum \(1000\) objective evaluations, relative parameter tolerance \(10^{-8}\), and relative objective tolerance \(10^{-10}\). The optimization is initialized either from the Fourier inversion estimator \(\hat\theta_T^{\mathrm{FI}}\) or from the robust matched-filter  estimator \(\hat\theta_T^{\mathrm{rob}}\). For an estimator \(\hat\theta_T=(\hat a_T,\hat\tau_T,\hat b_T)\), we report the relative mean squared error (MSE) through
\begin{align*}
\mathrm{Relative \ MSE}(a,\tau,b)
&=
\left(\frac{\hat a_T-a_0}{a_0}\right)^2\\
 &+ \left(\frac{d_K(\hat\tau_T,\tau_0)}{K}\right)^2
 + \left(\frac{\hat b_T-b_0}{b_0}\right)^2,
\end{align*}
where \(d_K\) denotes circular distance on the period-\(K\) circle. The plotted relative MSE is the Monte-Carlo average of this quantity. In \Cref{fig:nominal}, we plot \(T\Delta\) times this relative MSE, together with the Fisher-information lower bound
\begin{equation}
\label{eq:fisher_bound_compute}
\Tr \left(W\mathcal I_0^{-1}\right) \, ,
\qquad
W=\Delta \Diag(a_0^{-2},K^{-2},b_0^{-2}) \, .
\end{equation}
For comparison, we also plot our Fisher information evaluated at $t_d= 0$ ($\alpha_r \equiv 1$ in \eqref{eq:Irate_def}) in \Cref{fig:nominal}. This recovers the dead-time-free Fisher bound used in prior work~\cite{kitichotkul2023role, pediredla2018signal, gupta2019asynchronous, lu2013adaptive}, where Fisher-information expressions and lower bounds for unbiased estimators were derived only under the assumption of negligible dead time. For the present high-flux parameters, ignoring dead time lowers the bound by about \(75.22\%\), substantially underestimating the minimum attainable error. 

This discrepancy becomes small only in a much lower dead time or lower-flux regime. For example, keeping all other parameters fixed, we find numerically that making the dead-time-aware bound exceed the dead-time-free bound by less than \(10\%\) requires a dead time of less than approximately \(0.46\,\mathrm{ns}\) at the nominal flux. Alternatively, if the dead time is kept at its original value of \(50\,\mathrm{ns}\), the entire incident intensity must be scaled by about \(0.05\), replacing the nominal pulse and background parameters \((a_0,b_0)\) by \((0.05 a_0, 0.05 b_0)\), to obtain the same \(10\%\) gap.

\begin{figure}[!tbp]
\centering
\includegraphics[width=\mywidth]{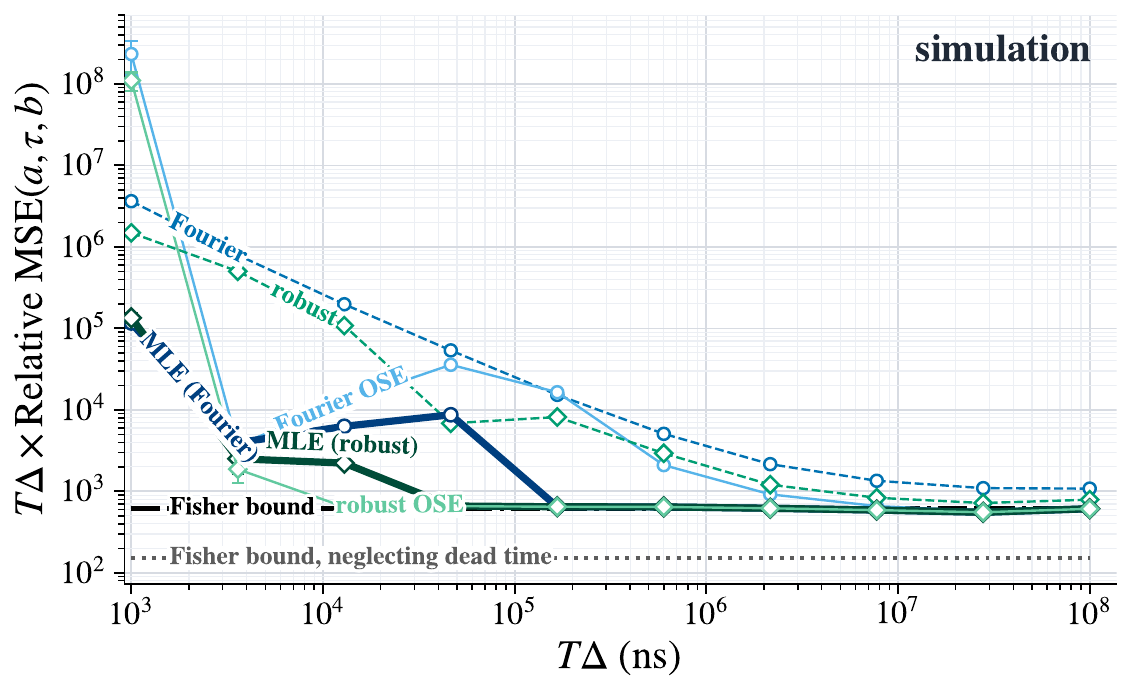}
\caption{
Mean-square error (MSE) for pulse parameters $(a,\tau,b)$  in single-photon lidar simulation scaled by observation time $T\Delta$, comparing various estimators: Fourier and robust pilot estimators (dashed lines), maximum-likelihood estimators (MLE) starting from pilots (thick solid lines), and one-step estimators (OSE) starting from pilots (thin solid lines).  Both MLE and OSE approach the Fisher-information lower bound (black line).  If dead time is neglected, the Fisher bound is lower (gray line).  Error bars (mostly too small to see) are standard error across $10^3$ repetitions.}
\label{fig:nominal}
\end{figure}

\begin{figure}[!tb]
\centering
\includegraphics[width=\mywidth]{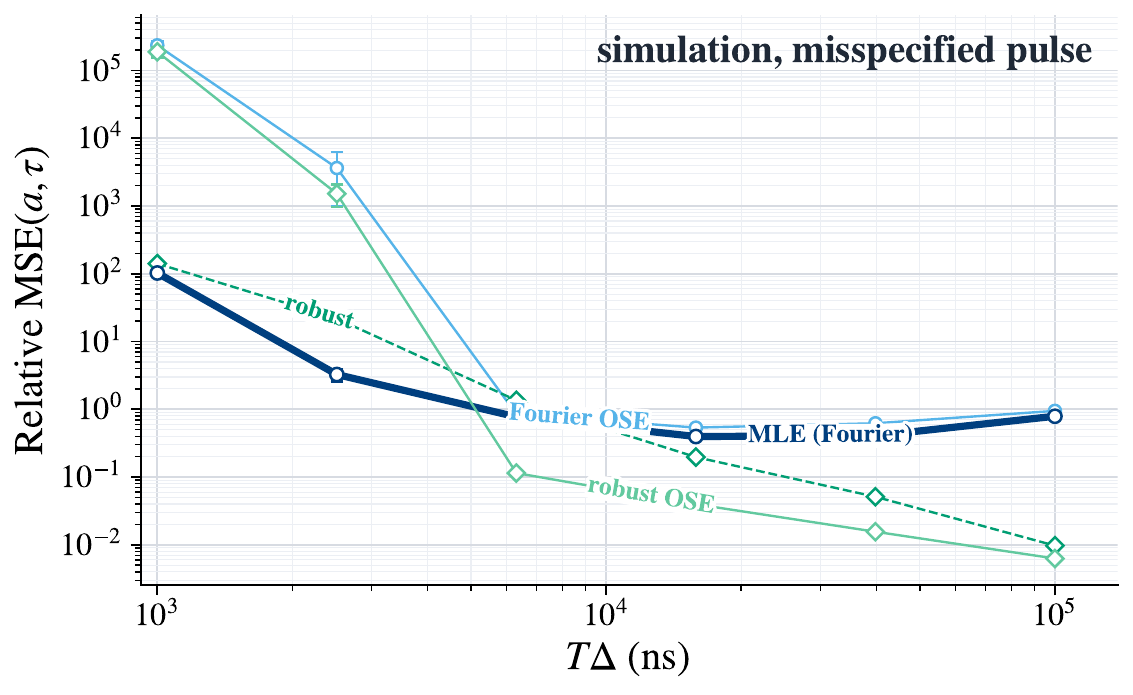}
\caption{Relative MSE for pulse parameters $(a,\tau)$ in the misspecified single-photon lidar simulation with an added localized background bump in \eqref{eq:lambda_mis}, while the fitted estimators still assume the constant-background model from \eqref{eq:lidar_lambda_model}. The Fourier-based estimators are strongly affected by the unmodeled bump, and the corresponding MLE/OSE can converge to poor local optima. By contrast, the robust pilot remains stable, and its one-step update gives the lowest error at moderate and large observation times. Error bars are standard error across \(10^3\) repetitions.}
\label{fig:robust}
\end{figure}

\begin{figure*}[!t]
\centering
\includegraphics[width=\textwidth]{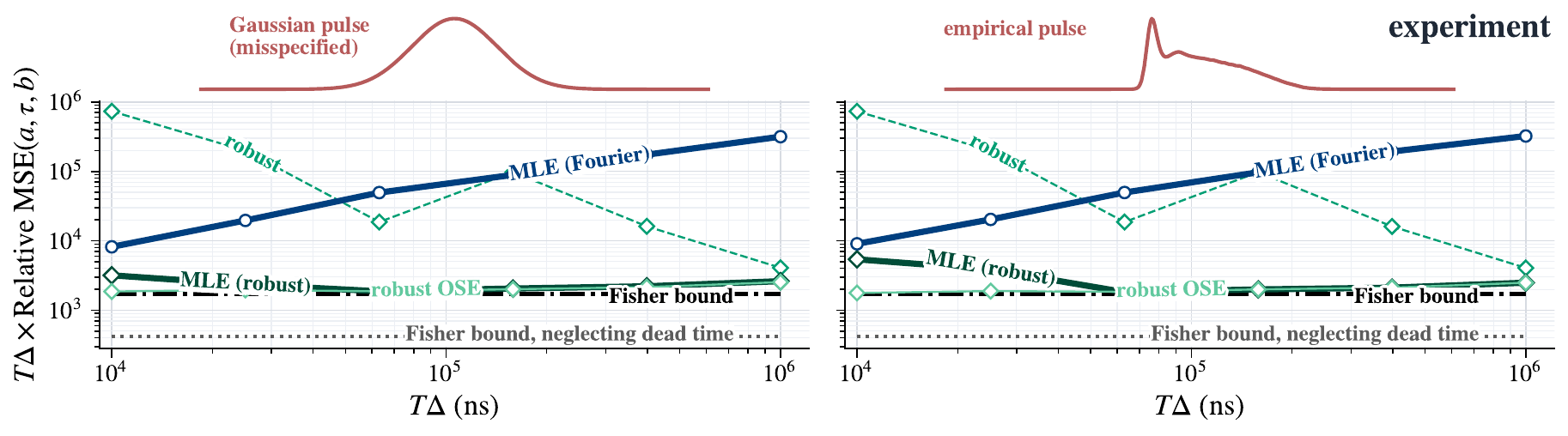}
\caption{Real single-photon lidar experiment~\cite{rapp2021high} comparing estimation under a misspecified wrapped-Gaussian pulse model (left) and an experimentally measured pulse model (right). The plotted quantity is \(T\Delta\)-scaled relative MSE for \((a,\tau,b)\). The likelihood optimizer initialized from the Fourier estimator fails to reach the efficient regime, whereas the robust-initialized likelihood optimizer and robust one-step estimator yield essentially the same result and remain close to the Fisher-information bound. If dead time is neglected, the Fisher bound is lower. The empirical pulse gives only modest improvement for the robust one-step estimator, indicating limited sensitivity to pulse-shape mismatch. The remaining gap to the bound therefore likely reflects residual experimental nonidealities (e.g., inhomogeneous background, timing jitter, detector effects) not captured by the idealized DED model. The MLE actually performs better under the misspecified wrapped-Gaussian pulse model at short acquisition times, but this difference disappears at longer horizons. Error bars are standard error across \(999\) blocks.}

\label{fig:real_data_highfluxspl}
\end{figure*}

The large-horizon behavior in \Cref{fig:nominal} agrees with the theory at practical time scales: by about \(T\Delta=10^5 \mathrm{ns}\), both MLE initializations approach the Fisher-information lower bound.  At intermediate horizons, however, the initialization matters substantially. For example, around \(T\Delta=5\times 10^4 \mathrm{ns}\) (which is a very relevant timescale from an applications perspective), the likelihood optimizer initialized from the Fourier-moment pilot is still far from the lower bound while  the robust initialization is already very close. This illustrates a practical weakness of the MLE: although it is asymptotically efficient as a statistical procedure, computing it requires a nonconvex numerical optimization, and poor initialization can lead to a bad local optimum, in particular at small sample sizes.

As explained in \Cref{subsection:one_step}, the one-step estimator avoids a global nonconvex likelihood optimization while retaining asymptotic efficiency. We now test this idea empirically in the same simulation setting as above, using the Fourier-inversion estimator and the robust matched-filter estimator from \Cref{subsection:lidar_estimators} as pilots $\tilde\theta_T$, followed by the one-step update  
$$
\hat\theta_T^{\mathrm{os}}
=
\tilde\theta_T
+
\bigl[\mathcal I(\tilde\theta_T;\hat\gamma_T)\bigr]^{-1}
\frac{U_T(\tilde\theta_T)}{T}
$$
from ~\eqref{eq:one_step_formula} with score $U_T(\tilde\theta_T)= \nabla \ell_T(\tilde\theta_T)$, log-likelihood $\ell_T$ as defined in \eqref{eq:mle} and \eqref{eq:phase_loglikelihood_intro}, Fisher information $\mathcal I(\tilde\theta_T;\hat\gamma_T)$ as in \eqref{eq:Irate_def}, and empirical gating frequencies $\hat \gamma_T$ as in \eqref{eq:empirical_gating_fractions}.

In \Cref{fig:nominal}, the empirical behavior is consistent with the $\sqrt{T}$-consistent pilot assumption of~\Cref{thm:onestep_efficient}, which we did not prove mathematically for either pilot. More importantly, the one-step updates recover the Fisher-information bound as predicted by the theory, yielding a substantial improvement
over the uncorrected pilot. The Fourier-inversion one-step becomes efficient very late around \(T\Delta=10^7 \mathrm{ns}\).  The robust matched-filter one-step, on the other hand, reaches the efficient plateau very quickly at about \(T\Delta=10^4 \mathrm{ns}\), which is faster than \textit{any} of the MLEs in \Cref{fig:nominal} without requiring the nonconvex optimization. This supports the practical role of the one-step estimators with good pilot estimators as an alternative to the MLE.

We next test robustness of these estimators  by  introducing a narrow localized background bump to the data-generating photon rate of the previous experiment
\begin{equation}
\label{eq:lambda_mis}
\lambda_r^{\mathrm{mis}}
=
a_0 f_{\tau_0}(r)+b_0+h u_{\tau_b}(r),
\end{equation}
where \(u_{\tau_b}(r)\) is the binned wrapped-Gaussian bump constructed from the Gaussian pulse in \eqref{eq:wrapped_gaussian} analogously to \(f_\tau(r)\), with parameters \(\sigma_b=1\) and \(\tau_b=70\), normalized to have unit peak, and we use $h = 0.15$.  Physically, this could correspond to an unmodeled reflection or inhomogeneous background counts. The fitted estimators, however, still use the previous model from equation \eqref{eq:lidar_lambda_model}. In \Cref{fig:robust}, we report the relative MSE only for the \((a,\tau)\)-coordinates, since under this misspecified data-generating model it is ambiguous what the correct constant-background parameter $b_0$ would be. 
The MLE initialized from the Fourier-inversion one-step estimate is strongly affected by the unmodeled bump as it converges to bad local minima in the misspecified setting.  By contrast, the robust matched-filter pilot remains stable because its median background estimate and template-matching step are less sensitive to localized deviations from the constant-background model. The one-step correction improves the robust pilot. This demonstrates our proposed  principle where estimators can be designed with considerations such as robustness in mind, without approximating the MLE, and then refined via a one-step update to convert a good global estimate into one with higher local accuracy and efficiency in the well-specified setting. We provide an additional validation using a real-world dataset to further support these conclusions empirically.

\subsection{Application to real data }
\label{sec:real_data_experiment}
We test the discussed estimators on real single-photon lidar data used in
\citeasnoun{rapp2021high} and posted online~\cite{highfluxspl_repo}. We refer the reader to that publication for a complete description of the experimental setup and summarize here the details important for our analysis. The physical laser period is approximately \(t_r=100 \mathrm{ns}\), which we aggregate to
\(K=625\) phase bins of width \(\Delta = 0.16 \mathrm{ns}\). The acquisition contains \(9{,}999{,}925\) laser periods. We split it into \(999\) non-overlapping DED trajectories that are \(10{,}000\) periods long. We use their single-pixel
ranging measurement \(\texttt{2019\_01\_28/FGS\_td198\_2019\_01\_28\_acq14.mat}\)
\cite[Fig.~5]{rapp2021high}, which contains the detection times from repeated laser returns in a high-flux free-running acquisition with physical dead time \(t_d = 198 \mathrm{ns}\), corresponding to approximately \(1238\) aggregated phase bins. For each block and horizon we compute the sufficient statistics \((N_r(T),S_r(T))_{r=0}^{K-1}\) exactly from the event stream and the known dead time. We estimate the proxy ground-truth parameter \(\theta_0=(a_0,\tau_0,b_0)\) directly from the full acq14 high-flux event stream by maximizing the empirical-pulse likelihood over \(a\), \(\tau\), and \(b\), using the experimentally measured laser pulse shape included in the data set. This gives the proxy $\theta_0=(a_0,\tau_0,b_0)
=
(0.2620,\ 17.933,\ 0.002469).$ As a diagnostic for the profile fit, we also compare the fitted intensity \(a_0 f_{\tau_0}(r)+b_0\) with the empirical intensity from \eqref{eq:pr_est} weighted by the bin-wise inverse variances, giving an excellent weighted fit with \(R^2=0.9915\). 

We now compare the MLE, the robust estimator, and the robust one-step estimator. To assess robustness, we compute each estimator under two fitted models: the empirically calibrated laser pulse for \(f\) in \eqref{eq:lidar_lambda_model}, and a misspecified wrapped-Gaussian pulse model as in \eqref{eq:wrapped_gaussian} whose second moment matches that of the calibrated pulse. This is a meaningful experiment because exact calibration data may be inaccessible in applications where environmental changes may affect the pulse shape.  Because the delay \(\tau\) is a continuous parameter, while the calibrated laser pulse is available only as a discrete histogram, we have to interpolate the calibrated pulse before evaluating likelihoods, scores, and Fisher information. After normalizing the pulse samples and forming their empirical CDF, we interpolate with monotone piecewise-cubic Hermite interpolating polynomials (PCHIP), which yield a continuous nonnegative density.\footnote{Since PCHIP interpolation is generally only \(C^1\), the calibrated-pulse experiment is an empirical diagnostic rather than a literal instance of \Cref{ass:lidar}. Also, we expect the \(C^3\) condition to be stronger than necessary, with differentiability in quadratic mean likely sufficient.}

The results in \Cref{fig:real_data_highfluxspl} are consistent with both our theory and the simulation results in \Cref{section:simulation}. We compute the Fisher-information lower bound using \eqref{eq:fisher_bound_compute}. For comparison, we also show the bound obtained under the approximation $t_d=0$, used in previous work \cite{kitichotkul2023role,pediredla2018signal,gupta2019asynchronous,lu2013adaptive}, which underestimates the minimum attainable error by \(75.58\%\). The robust one-step estimator and the MLE initialized from it closely track the Fisher-information prediction over practically relevant horizons, indicating that our asymptotic theory remains predictive in a realistic acquisition setting, despite unavoidable model mismatch and experimental noise. The MLE initialized from the Fourier-inversion estimator often converges to a poor local optimum and therefore fails to reach the efficient regime.

The data also illustrate the robustness gains of our proposed strategy of pairing the one-step estimator with a good pilot. On the one hand, when the MLE is initialized from the robust estimator, the calibrated-pulse likelihood incurs substantially larger risk at the two smallest horizons, by about \(70\%\) and \(64\%\), respectively. This appears to be a finite-sample overfitting effect, against which the misspecified Gaussian pulse acts as a smoother regularizing template. On the other hand, the one-step estimator does not exhibit this pathology. Across all six horizons, its calibrated-pulse risk is \(0.6\%\) to \(4.8\%\) lower than the Gaussian-pulse version and lower than the MLE, which is notable given how visibly the two pulses differ (\Cref{fig:real_data_highfluxspl}).

\begin{figure}[!tb]
\centering
\includegraphics[width=\mywidth]{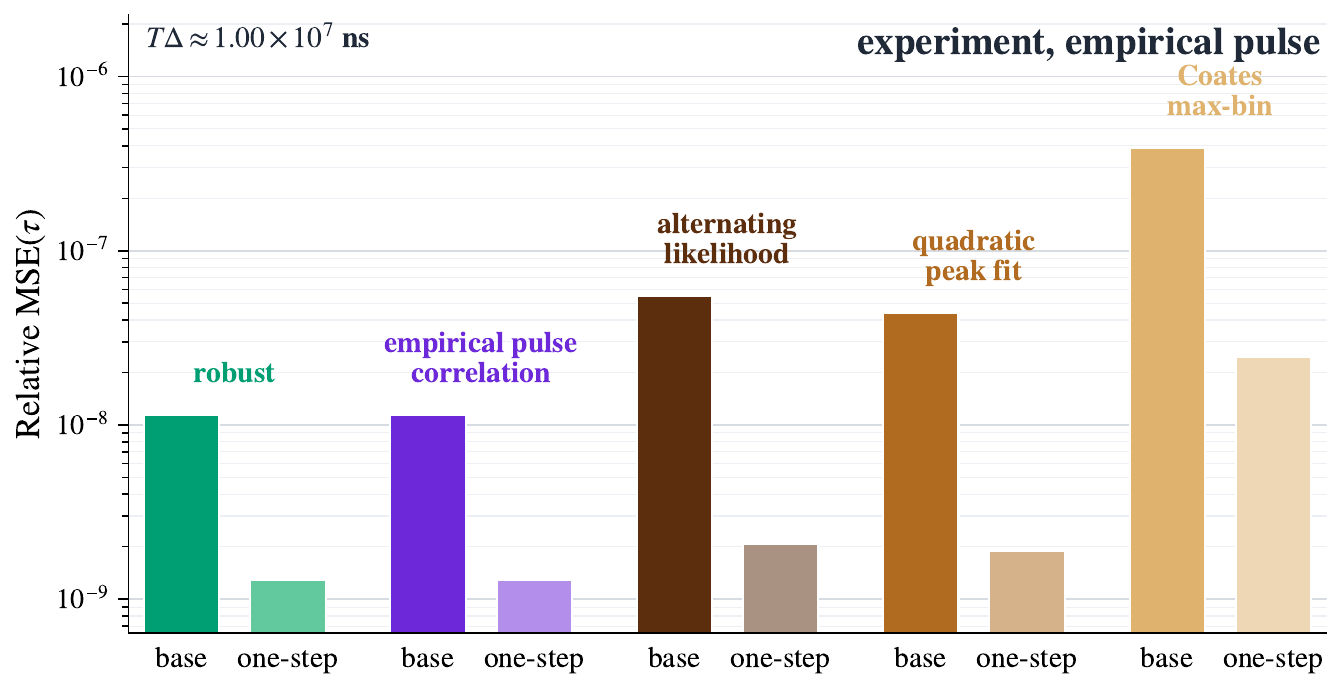}
\caption{
Relative mean-square error for the delay parameter \(\tau\) in a real single-photon lidar experiment~\cite{rapp2021high} at the rounded horizon \(T\Delta \approx 10^7\)~ns using the calibrated empirical pulse. Bars compare each range estimator before correction (``base'') and after one Fisher-scoring one-step update (``one-step'') for our robust pilot, empirical pulse correlation~\cite{mccarthy2013kilometer}, alternating likelihood~\cite{kitichotkul2025free}, quadratic peak fitting~\cite{pawlikowska2017single}, and Coates max-bin~\cite{gupta2019photon}. The one-step update substantially reduces the residual bias of the range-only estimators.}
\label{fig:ose_comparisons}
\end{figure}
Finally, we apply the one-step estimation idea to several estimators from the single-photon lidar imaging literature and compare them with the robust OSE proposed here. We use our implementations of the alternating-likelihood estimator of \citeasnoun{kitichotkul2025free}, empirical pulse correlation of \citeasnoun{mccarthy2013kilometer}, quadratic peak fitting of \citeasnoun{pawlikowska2017single}, and Coates max-bin of \citeasnoun{gupta2019photon}. Since the latter three methods are range-only estimators and do not estimate \((a,b)\) as required for one-step estimation \eqref{eq:one_step_formula}, we fill in these entries by fixing \(\tau\) at the reported range estimate and maximizing the likelihood in \eqref{eq:phase_loglikelihood_intro} over \((a,b)\) using L-BFGS before applying the one-step update. In contrast, our robust matched estimator provides a direct, optimization-free pilot for all three parameters. Although \citeasnoun{kitichotkul2025free} use a ``censoring-based'' estimator to initialize the optimization over \((a,\tau,b)\), our experiments show that the robust pilot is consistently one or more orders of magnitude more accurate than this initializer. We test the one-step updates for these estimators from the literature in the representative calibrated-pulse setting (right side of \Cref{fig:real_data_highfluxspl}), noting that we observed the same qualitative behavior described in this paragraph in the other scenarios (\Cref{fig:nominal}, \Cref{fig:robust}) as well. In \Cref{fig:ose_comparisons}, we evaluate the relative MSE of the prediction at $T\Delta = 10^7$~ns for the respective base estimator and our proposed one-step update, averaging over 100 non-overlapping trajectories. We observe a significant improvement in prediction performance in all of these cases. We briefly point out why the estimators considered from the literature exhibit such large errors when uncorrected. The empirical pulse correlation and quadratic peak fitting ignore the dead-time-induced histogram distortion, leading to some bias~\cite{pawlikowska2017single, mccarthy2013kilometer}.\footnote{Strictly speaking, these two effects prevent the pilots presented in the exposition from satisfying exact \(\sqrt{T}\)-consistency as stated in the asymptotic theory.}
All of these estimators incur bias because they estimate the peak on the discrete histogramming grid, meaning that they are resolution-limited by the grid. They typically compensate by retaining a high-resolution histogram (at the memory cost discussed in \Cref{subsection:likelihood_suff}) or by taking a few MLE update steps from this as a starting point~\cite{gupta2019photon, kitichotkul2025free, pawlikowska2017single, mccarthy2013kilometer}. Our one-step update is therefore an alternative to iterating MLE steps and complementary to grid refinement, while the robust matched pilot provides an optimization-free way to initialize this correction for the full parameter vector \((a,b,\tau)\).

\section{Conclusion}
\label{section:conclusion}
We introduced dead-time event detection (DED) processes to formalize periodic event detection under nonparalyzable dead time, a setting arising in several scientific domains, and developed the corresponding asymptotic theory. In particular, we established local asymptotic normality (LAN) for these models, from which the relevant Fisher-information lower bound and the asymptotic efficiency of both the MLE and one-step estimators follow. Empirically, our lidar experiments, including the real-data example, indicate that the theory is predictive at realistic sample sizes and remains useful under mild model misspecification. The experiments also show the practical appeal of one-step updates as an alternative to MLE optimization for DED processes. A one-step correction substantially improves a range of preliminary estimators, including both the  $O(K \log K)$ robust pilot proposed here and several commonly used estimators from the single-photon lidar literature, while requiring only $O(K)$ additional cost and no iterative nonlinear optimization.

Beyond lidar, it would be interesting to apply our theory and proposed one-step correction to other concrete examples,  such as fluorescence lifetime imaging~\cite{becker2004fluorescence, becker2012fluorescence} or photon rate reconstruction in  non-line-of-sight imaging~\cite{drost2015dead, rapp2020seeing, faccio2020non} and quantum vibration sensing~\cite{lualdi2026quantum, lualdi2025fast}. 

The theory developed here is also useful beyond estimation.  The derivation of  a sufficient statistic in \Cref{prop:suff_active} and the Fisher information  $\mathcal I_0$ in \eqref{eq:I0_def} may inform the design and evaluation of lossy data-compression schemes. Moreover, criteria such as minimizing \(\Tr(\mathcal I_0^{-1})\), as suggested by \Cref{thm:hajek_convolution,thm:local_asymptotic_minimax}, can guide the design of gating schemes by optimizing the limiting gating frequencies \(\gamma\) subject to dead-time constraints. We are currently preparing a manuscript addressing this topic. 

There are several directions for extending our theoretical results.  One is to treat paralyzable (Type II) dead time, where arrivals during the inactive period extend the dead time.  Another is statistical testing, which should follow naturally from the LAN formulation and is important in applications such as astrophysics and high-energy physics. Extending the model to multiple spatially coupled parameters or detector arrays would also provide theory for global multipixel reconstruction methods that exploit spatial structure~\cite{kirmani2014first,shin2015photon,liu2019pile, vicente2013improved, taguchi2011modeling, lualdi2026quantum, lualdi2025fast}. Other natural extensions would be to allow model memory beyond dead time (as in the modeling of neurons), as well as stochastic rate functions (which arise in both neuroscience and astrophysical applications) \cite{pillow2008spatio,keat2001predicting,mcintosh2016deep,berry1997refractoriness, huppenkothen2022accurate}.  It would also be useful to quantify the effect of model misspecification by deriving the limiting error rates of the MLE and one-step estimators under misspecified DED models, extending existing statistical theory for the misspecified setting~\cite{white1982maximum}. In particular, this would provide a theoretical understanding of how sources of error not explicitly accounted for in our model such as timing jitter and pulse shape misspecification \cite{buller2009single, incoronato2021statistical} affect estimation performance.  Finally, a continuous-time version of our  DED framework could connect the present theory to the statistical literature on Poisson processes~\cite{kutoyants2012statistical}, while extending it to the dead-time--censored single-photon setting considered here.

\section*{Acknowledgments}
The authors thank J.~Rapp, R.~Kitichotkul, and V.~K.~Goyal for helpful discussions and various pointers into the single-photon detection literature, and B.~Horwitz for discussion of neural analogues. ChatGPT was used for proofreading and minor language editing. OpenAI Codex was used to assist software development and figure generation. The authors reviewed all AI-assisted text and code and take full responsibility for the manuscript, computations, and figures. This work was supported in part by the Singapore--MIT Alliance for Research and Technology (SMART) Wafer-scale Integrated Sensing Devices based on Optoelectronic Metasurfaces (WISDOM) interdisciplinary research group, and by a grant from the Simons Foundation. 

\section*{Code availability}
The analysis and figure-generation code used to produce the plots in this paper is available  at \url{https://github.com/FredericJorgensen/dead-time-event-estimation}.

\appendices

\section{Gating-frequency convergence of acquisition schemes}
\label[appendix]{app:gating_frequency_convergence_acquisition}
We verify gating-frequency convergence for the free-running and synchronous acquisition schemes discussed in \Cref{subsection:observation_model}.
For a direct empirical comparison of these schemes we refer the reader to the single-photon detection literature~\cite{gupta2019photon, gupta2019asynchronous,antolovic2015nonuniformity, kitichotkul2025free}.
The argument is pointwise in $\theta$, so throughout this appendix we suppress $\theta$ and write $p_r\coloneq p_r(\theta)$. We assume $0<p_r<1$ for every phase $r$, as in \Cref{ass:lan}.

\begin{proposition}
\label{prop:gating_frequency_convergence_fr_syn}
The free-running and synchronous acquisition schemes in \Cref{subsection:observation_model} are gating-frequency convergent.
\end{proposition}

\begin{proof}
Define the dead-time timer $D_t\in\{0,\dots,D\}$ recursively by $D_0 = 0$ and 
$$
D_{t+1}\coloneq 
\begin{cases}
D, & Y_t=1\\
(D_t-1)_+ & Y_t=0 
\end{cases} \; ,
$$
where $D_t=0$ means that the detector is available. For free-running acquisition, $G_t=\Ind{D_t=0}.$
For synchronous acquisition, define the within-period detection indicator variable $W_t$, initialized by $W_0=0$ and updated by
$$
W_{t+1}\coloneq 
\begin{cases}
0, & (t+1)\bmod K=0\\
W_t\vee Y_t & \text{otherwise}
\end{cases} \; .
$$
Then, writing \(q(t)\coloneq t-(t\bmod K)\) for the beginning of the period containing \(t\), 
\[
G_t=\Ind{D_{q(t)}=0}\Ind{W_t=0}.
\]
We have $Y_t=G_t Z_t $ for $(Z_t)_{t\ge0}$ independent with $Z_t\sim\Ber(p_t)$ in either acquisition scheme.

It suffices to give one argument covering both cases. For either rule, the sampled process
$$
\bar D_\ell\coloneq D_{\ell K} \, ,\qquad \ell=0,1,\dots \; ,
$$
is a time-homogeneous Markov chain on the finite state space $\{0,\dots,D\}$, because the dynamics within each period are identical and the variables $(Z_t)$ are independent with $K$-periodic probabilities. 
Since \(p_r<1\) for every phase \(r\), every state can reach \(0\) with positive probability by observing no detections for sufficiently many periods. Hence every closed communicating class contains \(0\). Also, \(0\) has a positive self-loop, since no detection during one full period returns the chain to \(0\). Thus the reachable skeleton chain has a unique closed aperiodic recurrent class containing \(0\), and hence a unique stationary distribution \(\mu\) on that class. The finite-state Markov-chain ergodic theorem gives, for every bounded function $f$,

$$
\frac1L\sum_{\ell=0}^{L-1} f(\bar D_\ell)
\xrightarrow{\mathrm{a.s.}}
\sum_{i=0}^{D}\mu_i f(i) \, .
$$

Fix a phase $r$. For the acquisition rule under consideration, define
$$
f_r(i)\coloneq \E[G_r\mid D_0=i] \, ,
$$
where in the synchronous case we initialize $W_0=0$. Let
$$
\mathcal H_\ell\coloneq \sigma(X_0,\ldots,X_{\ell K-1},D_{\ell K})
$$
be the information available at the start of the period beginning at \(\ell K\).  By \(K\)-periodicity and the Markov property, and since in the synchronous case \(W_{\ell K}=0\) at the beginning of every period,
\[
\E[G_{r+\ell K}\mid \mathcal H_\ell]
=
f_r(\bar D_\ell).
\]
Decompose
$$
\frac1L\sum_{\ell=0}^{L-1} G_{r+\ell K}
=
\frac1L\sum_{\ell=0}^{L-1} f_r(\bar D_\ell)
+\frac1L\sum_{\ell=0}^{L-1}\bigl(G_{r+\ell K}-f_r(\bar D_\ell)\bigr) \, .
$$
The first term converges almost surely to $\sum_i\mu_i f_r(i)$. For the second term, note that 
$$
\E \left[G_{r+\ell K}-f_r(\bar D_\ell)\mid \mathcal H_\ell\right]=0. 
$$
If \(\ell<v\), the \(\ell\)-th summand is \(\mathcal H_v\)-measurable, so the cross terms in the variance computation vanish. The summands are bounded by $1$, and hence
$$
\Var \left(\frac1L\sum_{\ell=0}^{L-1}\bigl(G_{r+\ell K}-f_r(\bar D_\ell)\bigr)\right)
\le \frac1{L^2}\sum_{\ell=0}^{L-1} 1
=O(1/L) \, .
$$
Hence the second term is $o_P(1)$. Thus gating-frequency convergence holds with
$$
\gamma_r\coloneq \sum_{i=0}^{D}\mu_i f_r(i) \, .
$$
Since $r$ was arbitrary, this proves the claim for both acquisition schemes.
\end{proof}

\section{Derivation of the lidar pilot estimators}
\label[appendix]{app:lidar_pilot_derivations}
We give the heuristic derivations behind the two lidar pilot estimators introduced in \Cref{subsection:lidar_estimators}. Both begin from the rate estimates \((\hat\lambda_r)_{r=0}^{K-1}\) in \eqref{eq:pr_est}, for which the log-transform is chosen so that the lidar model becomes additive on the photon-rate scale:
\[
\hat\lambda_r \approx \lambda_r(\theta)=a f_\tau(r)+b.
\]

\subsection{Fourier inversion estimator}
The Fourier inversion estimator is an approximate method-of-moments estimator based on the lowest Fourier modes of \((\hat\lambda_r)_{r=0}^{K-1}\). Define
\[
\hat M_k
=
\frac1K
\sum_{r=0}^{K-1}
\hat\lambda_r \exp(2\pi i k r/K),
\qquad k=0,1.
\]
Using the approximation \(\E(\hat\lambda_r)\approx \lambda_r(\theta)\), the corresponding moment equations are
\[
\hat M_k
=
\frac1K
\sum_{r=0}^{K-1}
\lambda_r(\theta)\exp(2\pi i k r/K),
\qquad k=0,1.
\]
For \(k=0\), the normalization of the pulse gives
\[
\frac1K\sum_{r=0}^{K-1}\lambda_r(\theta)
=
\frac1K\sum_{r=0}^{K-1}
\bigl(a f_{\tau}(r)+b\bigr)
=
\frac{a}{K}+b.
\]

To evaluate the first nonzero Fourier mode, write the Fourier series of the periodic pulse as
\[
\begin{aligned}
f(x)
&=
\sum_{m\in\mathbb Z}
c_m \exp(2\pi i m x/K),
\\
c_m
&=
\frac1K
\int_0^K
f(x)\exp(-2\pi i m x/K) dx.
\end{aligned}
\]
A direct integration over one bin gives
\begin{equation}
\label{eq:app:binned_fourier_coefficients}
\begin{aligned}
f_{\tau}(r)
&=
\sum_{m\in\mathbb Z}
d_m
\exp \left(2\pi i m(\tau-r)/K\right),
\\
d_m
&=
c_m
\exp(\pi i m/K)
\operatorname{sinc} \left(\frac{m}{K}\right).
\end{aligned}
\end{equation}
Therefore
\[
\begin{aligned}
\frac1K
\sum_{r=0}^{K-1}
\lambda_r(\theta)\exp(2\pi i r/K)
&=
a\exp(2\pi i \tau/K)\\
&\quad{}\times
\sum_{\ell\in\mathbb Z}
d_{1+\ell K}\exp(2\pi i \ell \tau),
\end{aligned}
\]
since all Fourier modes except those congruent to \(1\) modulo \(K\) cancel in the discrete sum. We then use the approximation
\[
\sum_{\ell\in\mathbb Z}
d_{1+\ell K}\exp(2\pi i \ell \tau)
\approx d_1,
\]
which is accurate when the higher Fourier coefficients \(d_{1+\ell K}\), \(\ell\neq0\), decay rapidly, for example when \(f\) is smooth relative to the binning resolution. This gives the estimator-defining approximate moment equations
\[
\hat M_0
=
\frac{a}{K}+b,
\qquad
\hat M_1
=
a d_1\exp(2\pi i \tau/K).
\]
Solving these equations yields the Fourier inversion estimator in \eqref{eq:lidar_fourier_estimator}. The coefficient \(d_1\) depends only on the calibrated or modeled pulse template \(f\), and can therefore be precomputed.

\subsection{Robust matched estimator}
The robust matched estimator starts from the same approximate additive model
\[
\hat\lambda_r\approx a f_\tau(r)+b,
\]
but avoids reducing the profile to a single Fourier coefficient. Because the pulses used in lidar are narrow, most phases are background-dominated; this motivates the initial robust background estimate
\[
\hat b_T^{(0)}
=
\operatorname{median}_{r}\hat\lambda_r,
\qquad
J_r=\hat\lambda_r-\hat b_T^{(0)}.
\]
After centering, the approximate signal model is \(J_r\approx a f_\tau(r)\). For a fixed delay \(\tau\), the nonnegative least-squares estimate of \(a\) is
\[
\left(
\frac{
\sum_{r=0}^{K-1} J_r f_{\tau}(r)
}{
\sum_{r=0}^{K-1} f_{\tau}(r)^2
}
\right)_+.
\]
Thus, when the template norm is fixed or varies slowly with \(\tau\), estimating \(\tau\) reduces to finding the shift with largest matched-filter score as defined in \Cref{eq:robust_range} which  can be computed efficiently using an FFT.  The final formulas in \Cref{subsection:lidar_estimators} then project \(J\) onto \(f_{\hat\tau_T^{\mathrm{rob}}}\) to estimate \(a\), and re-estimate \(b\) as the median residual. The second median step makes the final background estimate robust to localized bins that remain poorly described by the shifted pulse template.

\section{Proofs of main results}
\label[appendix]{section:proofs}
\begin{proof}[Proof of \Cref{prop:lan}]
All convergences below are under $P_{\theta_0,t}$. The bound on the phasewise rates in \Cref{ass:lan} implies that there is $\varepsilon>0$ such that
$$
\varepsilon <p_t(\theta) < 1-\varepsilon
$$
for all $\theta \in \Theta$ and indices $t$.  Therefore, the finite-horizon laws $P_{\theta,t}$, $\theta\in \Theta$, have a common support consisting of the feasible paths compatible with the gating rule and the constraint $G_s=0\Rightarrow Y_s=0$ and the likelihood ratio is well-defined. 
Define for $r\in\N$, $y\in\{0,1\}$, and $\theta \in \Theta$
\begin{equation}
\label{eq:single_bin_derivatives}
\begin{aligned}
\varphi_r(y;\theta)&\coloneq p_r(\theta)^y(1-p_r(\theta))^{1-y},\\
m_r(y;\theta)&\coloneq \log \varphi_r(y;\theta),\\
s_r(y;\theta)&\coloneq \nabla_\theta m_r(y;\theta),\\
J_r(y;\theta)&\coloneq \nabla_\theta^2 m_r(y;\theta).
\end{aligned}
\end{equation}
Then,   the $\theta$-independent gating term cancels out in \eqref{eq:factorization} and we obtain,
$$
\log \frac{dP_{\theta,t}}{dP_{\theta_0,t}}
=
\sum_{r=0}^{t-1} G_r\Bigl(m_r(Y_r;\theta)-m_r(Y_r;\theta_0)\Bigr) \, .
$$
Since each $\lambda_r$ extends to a $C^3$ map on an open set containing \(\Theta\) and $p_r(\theta)\in[\varepsilon,1-\varepsilon]$, all first, second, and third derivatives of $m_r(y;\theta)$ are uniformly bounded for $r=0,\dots,K-1$, $y\in\{0,1\}$, and $\theta\in\Theta$. Hence, for bounded $h_t$ and with $\theta_t\coloneq \theta_0+h_t/\sqrt t$, Taylor's theorem
with uniform remainder gives
\begin{mymultline*}
m_r(Y_r;\theta_t)-m_r(Y_r;\theta_0) = \\
\frac1{\sqrt t} h_t^\top s_r(Y_r;\theta_0)
\mathrel{}+\frac1{2t} h_t^\top J_r(Y_r;\theta_0) h_t
+R_{r,t} \, ,
\end{mymultline*}
where differential operators refer to the parameter argument and
$$
|R_{r,t}|
\le C \frac{\|h_t\|^3}{t^{3/2}}
$$
for some constant $C<\infty$ independent of $r,t$.
Summing over $r=0,\dots,t-1$ yields
\begin{equation}
\label{eq:lan_taylor_step}
\Lambda_t(\theta_0+h_t/\sqrt t,\theta_0)
 = 
h_t^\top \Delta_t(\theta_0)
+\frac12 h_t^\top A_t h_t 
+o(1) \, ,
\end{equation}
where
$$
\begin{aligned}
\Delta_t(\theta_0)
&=
\frac1{\sqrt t}\sum_{r=0}^{t-1} G_r s_r(Y_r;\theta_0) \, ,
\\
A_t
&=
\frac1t\sum_{r=0}^{t-1} G_r J_r(Y_r;\theta_0) \, ,
\end{aligned}
$$
and the $o(1)$ is deterministic, hence also $o_{P_{\theta_0,t}}(1)$. We now identify this with the expression of the score in the statement.
Since
$$
\nabla p_r(\theta_0)=(1-p_r(\theta_0))\nabla \lambda_r(\theta_0) \, ,
$$
we have
$$
\begin{aligned}
s_r(y;\theta_0)
&=
\left(\frac{y}{p_r(\theta_0)}-\frac{1-y}{1-p_r(\theta_0)}\right)\nabla p_r(\theta_0)
\\
&=
\frac{y-p_r(\theta_0)}{p_r(\theta_0)} \nabla \lambda_r(\theta_0) \, .
\end{aligned}
$$
Thus, $\Delta_t(\theta_0)$ is exactly \eqref{eq:score_def_corrected}.
We bound the quadratic term in order to conclude the asymptotic formula for the likelihoods. Define the filtration
$$
\mathcal F_r\coloneq \sigma(X_0,\dots,X_r,G_{r+1}),\qquad r\ge 0 \, ,
$$
with $\mathcal F_{-1}\coloneq \sigma(G_0)$. Set
$$
\begin{aligned}
B_t
&\coloneq 
A_t+\frac1t\sum_{r=0}^{t-1}G_r\mathcal I_r(\theta_0)\\
&=
\frac1t\sum_{r=0}^{t-1}
\left(G_rJ_r(Y_r;\theta_0)+G_r\mathcal I_r(\theta_0)\right) \, .
\end{aligned}
$$
We claim that \(B_t\to 0\) in \(L^2\), hence \(B_t\xrightarrow{P_{\theta_0,t}}0\). Note first that for $Y\sim \Ber(p_r(\theta_0))$,
$$
\begin{aligned}
-\E_{\theta_0}[J_r(Y;\theta_0)]
&=
\E_{\theta_0}\!\left[
s_r(Y;\theta_0)s_r(Y;\theta_0)^\top
\right]\\
&=
\mathcal I_r(\theta_0) \, ,
\end{aligned}
$$
by \Cref{lem:second_moment_identity}. Consequently,
$$
\E_{\theta_0} \left[
G_rJ_r(Y_r;\theta_0)+G_r\mathcal I_r(\theta_0)
\mid \mathcal F_{r-1}
\right]
=0 \, .
$$
So the summands $Z_r
\coloneq 
G_rJ_r(Y_r;\theta_0)+G_r\mathcal I_r(\theta_0)$ in $B_t$ are martingale differences. $\|Z_r\|_{\text{F}}^2\leq C'$ is also a.s.~bounded by a constant $C'$ ($C^3$ assumption).
Now
\begin{align*}
\E_{\theta_0}\|B_t\|_{\text{F}}^2
&=
\frac1{t^2}\E_{\theta_0}\left\|\sum_{r=0}^{t-1} Z_r\right\|_{\text{F}}^2 \\
&=
\frac1{t^2}\sum_{r,v=0}^{t-1}\E_{\theta_0} \langle Z_r,Z_v\rangle_{\text{F}} \, .
\end{align*}
If \(r<v\), then \(Z_r\) is \(\mathcal F_{v-1}\)-measurable, and therefore by the tower-property of conditional expectation $\E_{\theta_0} \langle Z_r,Z_v\rangle_{\text{F}} = 0$. Hence all cross terms vanish, and so
\begin{align*}
\E_{\theta_0}\|B_t\|_{\text{F}}^2
&=
\frac1{t^2}\sum_{r=0}^{t-1}\E_{\theta_0}\|Z_r\|_{\text{F}}^2 \leq C'/t \, .
\end{align*}
Thus \(B_t\to 0\) in \(L^2\), hence \(B_t\xrightarrow{P_{\theta_0,t}}0\). Further, by gating-frequency convergence, with \(L=\lfloor t/K \rfloor\),
$$
\frac1t\sum\limits_{r=0}^{t-1}
G_r\mathcal I_r(\theta_0)  =  \mathcal{I}(\theta_0;\gamma) + o_{P_{\theta_0,t}}(1) \, .
$$
In particular, $A_t
=
-\mathcal I(\theta_0;\gamma)+o_{P_{\theta_0,t}}(1).$
Substituting this into \eqref{eq:lan_taylor_step} gives
\begin{mymultline*}
\Lambda_t(\theta_0+h_t/\sqrt t,\theta_0) = \\
h_t^\top\Delta_t(\theta_0)
-\frac12 h_t^\top \mathcal I(\theta_0;\gamma) h_t
+o_{P_{\theta_0,t}}(1) \, ,
\end{mymultline*}
which is exactly \eqref{eq:lan_expansion}.

We now prove that $\Delta_t(\theta_0)\Rightarrow \Normal(0,\mathcal I(\theta_0;\gamma))$. By the   Cram\'er--Wold theorem it is sufficient (in fact, necessary) to prove $u^\top\Delta_t(\theta_0)\Rightarrow \Normal(0,u^\top \mathcal I(\theta_0;\gamma)u)$ for all nonzero $u\in\R^d$; the case $u=0$ is trivial.  Define
$$
D_r(u)\coloneq u^\top G_r s_r(Y_r;\theta_0) \, ,
\qquad r=0,\dots,t-1 \, .
$$
Then
$$
u^\top \Delta_t(\theta_0)=\frac1{\sqrt t}\sum_{r=0}^{t-1}D_r(u) \, .
$$
Set
$$
\sigma_u^2\coloneq u^\top \mathcal I(\theta_0;\gamma)u \, .
$$
By positive definiteness of \(\mathcal I(\theta_0;\gamma)\), \(\sigma_u^2>0\). 
For each $t$, define a scalar martingale-difference array by
$$
\xi_{t,r}\coloneq \frac{D_r(u)}{\sqrt t \sigma_u} \, ,
\qquad
\mathcal G_{t,r}\coloneq \mathcal F_r,
\qquad r=0,\dots,t-1 \, .
$$
Since $G_r$ is $\mathcal F_{r-1}$-measurable and
$$
\E_{\theta_0} \left[s_r(Y_r;\theta_0)\mid \mathcal F_{r-1}\right]=0
\quad\text{on }\{G_r=1\} \, ,
$$
while $D_r(u)=0$ on $\{G_r=0\}$, we have
$$
\E_{\theta_0}[\xi_{t,r}\mid \mathcal G_{t,r-1}]=0 \, .
$$
Thus $\{\xi_{t,r},\mathcal G_{t,r}\}_{r=0}^{t-1}$ is a martingale-difference array. Next, its conditional variance satisfies
\begin{align*}
V_t^2
&\coloneq 
\sum_{r=0}^{t-1}\E_{\theta_0}[\xi_{t,r}^2\mid \mathcal G_{t,r-1}] \\
&=
\frac{1}{t \sigma_u^2}\sum_{r=0}^{t-1}\E_{\theta_0}[D_r(u)^2\mid \mathcal F_{r-1}] \\
&=
\frac{1}{t \sigma_u^2}
u^\top \left(\sum_{r=0}^{t-1}G_r\mathcal I_r(\theta_0)\right) u
\xrightarrow{P_{\theta_0,t}} 1 \, ,
\end{align*}
by our shown convergence of the empirical Fisher information. Finally, since $p_r(\theta_0)\in[\varepsilon,1-\varepsilon]$ and the phase set is finite, there exists $C_u<\infty$ such that
$$
|D_r(u)|\le C_u
\qquad\text{a.s.~for all }r \, .
$$
Hence, for every $\delta>0$,
$$
|\xi_{t,r}|
\le \frac{C_u}{\sqrt t \sigma_u}
<\delta
$$
for all sufficiently large $t$, uniformly in \(r\). Therefore
$$
\sum_{r=0}^{t-1}
\E_{\theta_0} \left[
\xi_{t,r}^2\Ind{|\xi_{t,r}|>\delta}
\mid \mathcal G_{t,r-1}
\right]
=0
$$
for all sufficiently large $t$, so the conditional Lindeberg condition holds.
By the scalar martingale central limit theorem~\cite[Theorem~4]{lalley2014martingale}
$$
\frac{u^\top \Delta_t(\theta_0)}{\sigma_u} = \sum_{r=0}^{t-1}\xi_{t,r}\Rightarrow \Normal(0,1) \, ,
$$
which shows the claimed normality result.

We conclude the proof by proving contiguity. To prove contiguity, take any subsequence.
It has a further subsequence along which $h_t\to h\in\R^d$ (we avoid relabeling).
Along that subsubsequence, by Slutsky's theorem
$$
\begin{aligned}
&\Lambda_t(\theta_0+h_t/\sqrt t,\theta_0)\\
&\quad\Rightarrow
h^\top Z-\frac12 h^\top \mathcal I(\theta_0;\gamma)h \, ,
\\
&\qquad
Z\sim \Normal(0,\mathcal I(\theta_0;\gamma)) \, .
\end{aligned}
$$
By the continuous mapping theorem,
$$
\begin{aligned}
&\exp \Bigl(\Lambda_t(\theta_0+h_t/\sqrt t,\theta_0)\Bigr)\\
&\quad\Rightarrow
\exp \Bigl(h^\top Z-\frac12 h^\top \mathcal I(\theta_0;\gamma)h\Bigr) \, .
\end{aligned}
$$
This limit has expectation $1$ and is almost surely strictly positive. Therefore, by Le Cam's first lemma, $P_{\theta_0+h_t/\sqrt t,t}$ and $P_{\theta_0,t}$ are mutually contiguous. 
\end{proof}
\begin{lemma}
\label{lem:second_moment_identity}
Let $Y\sim \Ber(p_r(\theta))$ in the setting of \Cref{prop:lan}. With the notation from \eqref{eq:single_bin_derivatives},
we have
$$
-\E_\theta[J_r(Y;\theta)]
=
\E_\theta \left[s_r(Y;\theta)s_r(Y;\theta)^\top\right]
=
\mathcal I_r(\theta) \, .
$$
\end{lemma}

\begin{proof}
Since
$$
\nabla m_r(y;\theta)
=
\frac{y-p_r(\theta)}{p_r(\theta)(1-p_r(\theta))} \nabla p_r(\theta) \, ,
$$
we obtain
\begin{align*}
\E_\theta&
\left[\nabla m_r(Y;\theta)\nabla m_r(Y;\theta)^\top\right]\\
&=
\frac{\mathrm{Var}_\theta(Y)}{p_r(\theta)^2(1-p_r(\theta))^2}
\nabla p_r(\theta)\nabla p_r(\theta)^\top \\
&=
\frac{1}{p_r(\theta)(1-p_r(\theta))}
\nabla p_r(\theta)\nabla p_r(\theta)^\top \, .
\end{align*}
On the other hand, a direct differentiation of
$$
m_r(y;\theta)=y\log p_r(\theta)+(1-y)\log(1-p_r(\theta))
$$
and taking expectations using $\E_\theta[Y]=p_r(\theta)$ gives
$$
-\E_\theta[\nabla^2 m_r(Y;\theta)]
=
\frac{1}{p_r(\theta)(1-p_r(\theta))}
\nabla p_r(\theta)\nabla p_r(\theta)^\top \, .
$$
Finally, since
$$
\nabla p_r(\theta)=(1-p_r(\theta))\nabla \lambda_r(\theta) \, ,
$$
both expressions equal
$$
\frac{1-p_r(\theta)}{p_r(\theta)}
\nabla \lambda_r(\theta)\nabla \lambda_r(\theta)^\top
=
\mathcal I_r(\theta) \, .
$$
\end{proof}
\begin{proof}[Proof of \Cref{thm:hajek_convolution}]
Define the estimator of the local parameter $\hat h_T \coloneq  \sqrt{T}(\hat\theta_T - \theta_0).$ By regularity, $\hat h_T - h \Rightarrow \mathcal L_{\theta_0}$ under $P_{\theta_0+h/\sqrt T,T}$, independently of $h$.
By \Cref{prop:lan} and the H\'ajek convolution theorem~\cite[Theorem~3]{le2000asymptotics}, we therefore have
$$
\mathcal L_{\theta_0}
=
\Normal\left(0,\mathcal I_0^{-1}\right) * \nu_{\theta_0}
$$
for some probability measure $\nu_{\theta_0}$ on $\R^d$.
\end{proof}

\begin{proof}[Proof of \Cref{thm:local_asymptotic_minimax}]
Given an estimator \(\hat\theta_T\), pass to the local parametrization
$$
h=\sqrt T(\theta-\theta_0),
\qquad
\hat h_T=\sqrt T(\hat\theta_T-\theta_0) \, .
$$
Then
$$
\sqrt T(\hat\theta_T-\theta)=\hat h_T-h \, .
$$
By \Cref{prop:lan}, the local experiments around \(\theta_0\) converge in the LAN sense with information matrix \(\mathcal I_0\).  Therefore the hypotheses of the local asymptotic minimax theorem~\cite[Section~6.6, Theorem~1]{le2000asymptotics} are satisfied.  Applying that theorem to the local estimator \(\hat h_T\) and substituting back \(h=\sqrt T(\theta-\theta_0)\) gives \eqref{eq:laminimax_general}.
\end{proof}

\begin{proof}[Proof of \Cref{thm:mle_efficient}]
Using the phase counts $N_r(T),S_r(T)$ from \Cref{subsection:likelihood_suff} and \eqref{eq:phase_counts}, we have
\begin{mymultline*}
\ell_T(\theta)
=
\sum_{r=0}^{K-1}
\Bigl[
S_r(T)\log p_r(\theta)\\
+
\bigl(N_r(T)-S_r(T)\bigr)\log\bigl(1-p_r(\theta)\bigr)
\Bigr] \, .
\end{mymultline*}
We can ignore the $\theta$-independent part here because it has no impact on the maxima of $\ell_T(\theta)$.
Set $p_r^0\coloneq p_r(\theta_0)$ and define
$$
M_r(T)
\coloneq 
S_r(T)-p_r^0N_r(T)
=
\sum_{\substack{0\le t<T\\ t\bmod K=r}}G_t\bigl(Y_t-p_r^0\bigr) \, .
$$
Because $G_t$ is predictable and, for \(t\bmod K=r\), conditional on $G_t=1$ we have $Y_t\sim\Ber(p_r^0)$,
the increments of $M_r(T)$ are bounded martingale differences with respect to our filtration $\mathcal F_t$ from before. Hence cross-terms vanish and we can expand the diagonal terms as
$$
\E_{\theta_0}M_r(T)^2
=
\sum_{\substack{0\le t<T\\ t\bmod K=r}}\E_{\theta_0} \left[G_t p_r^0(1-p_r^0)\right]
\le T \, ,
$$
so $M_r(T)/T\to 0$ in $L^2$, hence $M_r(T)/T\xrightarrow{P_{\theta_0,T}}0$. Together with gating-frequency convergence,

\begin{gather*}
\frac{S_r(T)}{T}\xrightarrow{P_{\theta_0,T}}\frac{\gamma_r(\theta_0)p_r^0}{K} \, ,
\\
\frac{N_r(T)-S_r(T)}{T}\xrightarrow{P_{\theta_0,T}}\frac{\gamma_r(\theta_0)(1-p_r^0)}{K} \, .
\end{gather*}
Since $\log p_r(\theta)$ and $\log(1-p_r(\theta))$ are uniformly bounded on \(\Theta\) by Condition (1) in \Cref{ass:lan}, we obtain
$$
\sup_{\theta\in\Theta}\left|\frac{\ell_T(\theta)}{T}-m(\theta)\right|
\xrightarrow{P_{\theta_0,T}}0 \, ,
$$
where
$$
m(\theta)
\coloneq 
\frac1K\sum_{r=0}^{K-1}\gamma_r(\theta_0)
\Bigl[
p_r^0\log p_r(\theta)+(1-p_r^0)\log\bigl(1-p_r(\theta)\bigr)
\Bigr] \, .
$$
Moreover, defining $R_+ \coloneq  \{  r : \gamma_r(\theta_0) > 0  \}$, we have
$$
\begin{aligned}
m(\theta)-m(\theta_0)
&=
-\frac1K\sum_{r\in R_+}\gamma_r(\theta_0)\\
&\qquad\qquad
\cdot D_{\mathrm{KL}} \left(\Ber(p_r^0) \middle\| \Ber(p_r(\theta))\right)\\
&\le 0 \, ,
\end{aligned}
$$
with equality if and only if $p_r(\theta)=p_r(\theta_0)$ for all $r\in R_+$ by (3) in \Cref{ass:lan} and injectivity of $\lambda \mapsto 1 - \exp(-\lambda)$.  In particular, $\theta_0$ is the unique maximizer of $m$ and by continuity of $m$ and compactness of $\Theta$, \(\sup\limits_{\theta\in \Theta: \|\theta -  \theta_{0}\|_2 \geq \varepsilon} m(\theta)<m\left(\theta_{0}\right)\) for every $\varepsilon > 0$. Therefore, Theorem~5.7 of \citeasnoun{van2000asymptotic} yields
$$
\hat\theta_T\xrightarrow{P_{\theta_0,T}}\theta_0 \, .
$$
This shows consistency.

Next, to show asymptotic normality, use the notation introduced in \eqref{eq:single_bin_derivatives}. Then
$$
\begin{aligned}
\nabla \ell_T(\theta)
&=\sum_{t=0}^{T-1}G_t s_t(Y_t;\theta) \, ,\\
\nabla^2 \ell_T(\theta)
&=\sum_{t=0}^{T-1}G_t J_t(Y_t;\theta) \, .
\end{aligned}
$$
At $\theta_0$,
$$
T^{-1/2}\nabla \ell_T(\theta_0)=\Delta_T(\theta_0) \, .
$$
By \Cref{prop:lan},
$$
\Delta_T(\theta_0)\Rightarrow \Normal(0,\mathcal I_0) \, .
$$
Work on the event that \(\hat\theta_T\) and the line segment from \(\hat\theta_T\) to \(\theta_0\) lie in the interior of \(\Theta\). 
Because $\theta_0$ is an interior point and $\hat\theta_T$ is consistent, this event has probability tending to one. Next, define
$$
\bar J_T
\coloneq 
-\frac1T\int_0^1
\nabla^2\ell_T \bigl(\theta_0+u(\hat\theta_T-\theta_0)\bigr) du \, .
$$
By the fundamental theorem of calculus,
\begin{equation}
\label{eq:fund_calc}
\nabla \ell_T(\hat\theta_T)
=
\nabla \ell_T(\theta_0)-T\bar J_T(\hat\theta_T-\theta_0) \, .
\end{equation}
Note that by the phasewise rates extending to $C^3$ maps on an open set containing \(\Theta\), $J_t(y;\cdot)$ is Lipschitz on \(\Theta\), uniformly in $t$ and $y\in\{0,1\}$, with a constant $L_\Theta$. Set $\delta_T\coloneq \hat\theta_T-\theta_0$. Applying this fact and the triangle inequality gives
\begin{align*}
\left\|\bar J_T+T^{-1}\nabla^2\ell_T(\theta_0)\right\|_{\text{F}}
&\leq T^{-1} \int_0^1
\Bigl\|\nabla^2\ell_T \bigl(\theta_0+u\delta_T\bigr)\\
&\qquad\qquad\qquad
{} - \nabla^2\ell_T(\theta_0) \Bigr\|_{\text{F}}du\\
&\leq T^{-1}\sum_{t=0}^{T-1}\int_0^1
G_t L_\Theta u\left\|\delta_T\right\|_2du\\
&\leq {L_\Theta} \int_0^1u\left\|\delta_T\right\|_2du\\
&\leq {L_\Theta}\left\|\delta_T\right\|_2 \, .
\end{align*}
Together with \(\delta_T=o_{P_{\theta_0,T}}(1)\), this bound implies
\[
\left\|\bar J_T+T^{-1}\nabla^2\ell_T(\theta_0)\right\|_{\text{F}}=o_{P_{\theta_0,T}}(1) \, .
\]
As shown in the proof of \Cref{prop:lan}, $-\frac{1}{T}\nabla^2\ell_T(\theta_0) = \mathcal I_0 + o_{P_{\theta_0,T}}(1)$, and therefore we can conclude that
$$
\bar J_T= \mathcal I_0 +  o_{P_{\theta_0,T}}(1) \, .
$$
The first-order condition for the interior maximizer gives \(\nabla\ell_T(\hat\theta_T)=0\). Combining this with  \eqref{eq:fund_calc}, we get
$$
0=\nabla\ell_T(\theta_0)-T\bar J_T(\hat\theta_T-\theta_0) \, .
$$
Multiplying by $1/\sqrt{T}$ implies that
\begin{equation}
\label{eq:rescaled_taylor}
o_{P_{\theta_0,T}}(1)=\Delta_T(\theta_0)-\bar J_T \sqrt{T}(\hat\theta_T-\theta_0).
\end{equation}
By $\bar J_T= \mathcal I_0 +  o_{P_{\theta_0,T}}(1)$, $\bar J_T$ is invertible with probability tending to $1$, and we can invert \eqref{eq:rescaled_taylor}  to get
\begin{equation}
\label{eq:linearization_final}
\sqrt{T} (\hat\theta_T-\theta_0)
=\bar J_T^{-1}\Delta_T(\theta_0) +  o_{P_{\theta_0,T}}(1).
\end{equation}
Now, Slutsky's theorem gives
$$
\sqrt{T} (\hat\theta_T-\theta_0)\Rightarrow \Normal(0,\mathcal I_0^{-1}) \, .
$$
Finally, we show regularity using Le Cam's third lemma and \eqref{eq:linearization_final}. Fix $h\in\R^d$.
By \Cref{prop:lan} we also have that
$$
\log\frac{dP_{\theta_0+h/\sqrt T,T}}{dP_{\theta_0,T}}
=
h^\top\Delta_T(\theta_0)
-\frac12 h^\top\mathcal I_0 h
+
o_{P_{\theta_0,T}}(1) \, .
$$
Since $\Delta_T(\theta_0)\Rightarrow \Normal(0,\mathcal I_0)$ under $P_{\theta_0,T}$, by Slutsky we get the joint convergence
\begin{gather*}
\begin{pmatrix}
\Delta_T(\theta_0)\\
\log\frac{dP_{\theta_0+h/\sqrt T,T}}{dP_{\theta_0,T}}
\end{pmatrix} \Rightarrow
\Normal\left(\mu_h,\Sigma_h\right) \, ,\\
\mu_h
\coloneq 
\begin{pmatrix}
0\\[1mm]
-\frac12 h^\top \mathcal I_0 h
\end{pmatrix} \, ,
\\
\Sigma_h
\coloneq 
\begin{pmatrix}
\mathcal I_0 & \mathcal I_0 h\\[1mm]
h^\top \mathcal I_0 & h^\top \mathcal I_0 h
\end{pmatrix}
\end{gather*}
under $P_{\theta_0,T}$. Le Cam's third lemma now yields
$$
\Delta_T(\theta_0)\Rightarrow
\Normal(\mathcal I_0 h,\mathcal I_0)
$$
under $P_{\theta_0+h/\sqrt T,T}$. We want to substitute this into the linearization \eqref{eq:linearization_final}. Note that by contiguity, the $o_{P_{\theta_0,T}}(1)$ becomes $o_{P_{\theta_0 + h/\sqrt{T},T}}(1)$ and therefore \(\bar J_T\) also has the same limit. In equations, this means that
$$
\sqrt T\left(\hat\theta_T-\theta_0-\frac{h}{\sqrt T}\right)
=
\mathcal I_0^{-1}\Delta_T(\theta_0)-h
+
o_{P_{\theta_0+h/\sqrt T,T}}(1) \, ,
$$
where we also  subtracted $h$ on both sides.  Slutsky's theorem concludes the proof
$$
\sqrt T
\left(
\hat\theta_T-\theta_0-\frac{h}{\sqrt T}
\right)
\Rightarrow
\Normal(0,\mathcal I_0^{-1}) \, .
$$
\end{proof}

\begin{proof}[Proof of \Cref{thm:onestep_efficient}]
Since $\sqrt T(\tilde\theta_T-\theta_0)=O_{P_{\theta_0,T}}(1)$, we have $\tilde\theta_T\xrightarrow{P_{\theta_0,T}}\theta_0$.
Because each $\lambda_r$ extends to a $C^3$ map on an open set containing \(\Theta\) and is bounded away from $0$ and $\infty$ on \(\Theta\), the maps
$\theta\mapsto \mathcal I_r(\theta)$ are $L$-Lipschitz on \(\Theta\), with constant $L > 0$ uniformly in $r=0,\dots,K-1$.
First we show that the empirical Fisher at the pilot is consistent.
\begin{mymultline*}
\bigl\|
\mathcal I \bigl(\tilde\theta_T;\hat\gamma_T\bigr)
-
\mathcal I \bigl(\theta_0;\hat\gamma_T\bigr)
\bigr\|
\\
 \le
\frac1K\sum_{r=0}^{K-1}\hat\gamma_{r,T} L\|\tilde\theta_T-\theta_0\|
\le
L\|\tilde\theta_T-\theta_0\| \, ,
\end{mymultline*}
and hence,
$$
\mathcal I \bigl(\tilde\theta_T;\hat\gamma_T\bigr)
=
\mathcal I \bigl(\theta_0;\hat\gamma_T\bigr)+o_{P_{\theta_0,T}}(1) \, .
$$
By gating-frequency convergence, $\mathcal I \bigl(\theta_0;\hat\gamma_T\bigr) =\mathcal I_0+o_{P_{\theta_0,T}}(1)$ and, therefore, 
$$
\mathcal I \bigl(\tilde\theta_T;\hat\gamma_T\bigr)=\mathcal I_0+o_{P_{\theta_0,T}}(1) \, .
$$
Since $\mathcal I_0$ is positive definite, $\mathcal I \bigl(\tilde\theta_T;\hat\gamma_T\bigr)$ is invertible with probability tending to $1$.
Using the exact same argument as the one following \eqref{eq:fund_calc} in the proof of \Cref{thm:mle_efficient} and defining $\bar J_T$ with $\tilde \theta_T$ instead of $\hat \theta_T$, we get the expansion
$$
U_T(\tilde\theta_T)
=
U_T(\theta_0)-T\bar J_T(\tilde\theta_T-\theta_0)
$$
with $\bar J_T = \mathcal I_0 +o_{P_{\theta_0,T}}(1)$.  Since $T^{-1/2}U_T(\theta_0)=\Delta_T(\theta_0)$, this becomes
$$
\frac1{\sqrt T}U_T(\tilde\theta_T)
=
\Delta_T(\theta_0)-\bar J_T\sqrt T(\tilde\theta_T-\theta_0) \, .
$$
Substituting into the definition of $\hat\theta_T^{\mathrm{os}}$,
\begin{align*}
\sqrt T(\hat\theta_T^{\mathrm{os}}-\theta_0)
&=
\sqrt T(\tilde\theta_T-\theta_0)
\\
&\quad+
\bigl[\mathcal I(\tilde\theta_T;\hat\gamma_T)\bigr]^{-1}
\frac1{\sqrt T}U_T(\tilde\theta_T) \\
&=
\bigl[\mathcal I(\tilde\theta_T;\hat\gamma_T)\bigr]^{-1}
\Delta_T(\theta_0)
\\
&\quad {} +
\Bigl[
I_d-
\bigl[\mathcal I(\tilde\theta_T;\hat\gamma_T)\bigr]^{-1} \times \bar J_T
\Bigr]\\
&\qquad\qquad{}\times
\sqrt T(\tilde\theta_T-\theta_0) \, .
\end{align*}
Now \Cref{prop:lan} gives $\Delta_T(\theta_0)=O_{P_{\theta_0,T}}(1)$, the assumption gives
$\sqrt T(\tilde\theta_T-\theta_0)=O_{P_{\theta_0,T}}(1)$, and we already proved
\begin{gather*}
\bigl[\mathcal I(\tilde\theta_T;\hat\gamma_T)\bigr]^{-1}
=
\mathcal I_0^{-1}+o_{P_{\theta_0,T}}(1) \, ,
\\
I_d-\bigl[\mathcal I(\tilde\theta_T;\hat\gamma_T)\bigr]^{-1}\bar J_T
=
o_{P_{\theta_0,T}}(1) \, .
\end{gather*}
Hence
$$
\sqrt T(\hat\theta_T^{\mathrm{os}}-\theta_0)
=
\mathcal I_0^{-1}\Delta_T(\theta_0)+o_{P_{\theta_0,T}}(1) \, .
$$
\Cref{prop:lan} further yields
$$
\Delta_T(\theta_0)\Rightarrow \Normal(0,\mathcal I_0) \, ,
$$
so Slutsky's theorem gives
$$
\sqrt T(\hat\theta_T^{\mathrm{os}}-\theta_0)\Rightarrow \Normal(0,\mathcal I_0^{-1}) \, .
$$
Fix now $h\in\R^d$. By \Cref{prop:lan}, $P_{\theta_0+h/\sqrt{T},T}$ is contiguous with respect to $P_{\theta_0,T}$.
Therefore every $o_{P_{\theta_0,T}}(1)$ term above is also $o_{P_{\theta_0+h/\sqrt{T},T}}(1)$, and so
$$
\sqrt T(\hat\theta_T^{\mathrm{os}}-\theta_0)
=
\mathcal I_0^{-1}\Delta_T(\theta_0)+o_{P_{\theta_0+h/\sqrt{T},T}}(1) \, .
$$
As in the proof of \Cref{thm:mle_efficient}, by \Cref{prop:lan} and Le Cam's third lemma,
$$
\Delta_T(\theta_0)\Rightarrow \Normal(\mathcal I_0 h,\mathcal I_0)
\qquad\text{under }P_{\theta_0+h/\sqrt{T},T} \, .
$$
Therefore,
$$
\begin{aligned}
\sqrt T\left(\hat\theta_T^{\mathrm{os}}-\theta_0-\frac{h}{\sqrt{T}}\right)
&=
\mathcal I_0^{-1}\Delta_T(\theta_0)-h+o_{P_{\theta_0+h/\sqrt{T},T}}(1)
\\
&\Rightarrow \Normal(0,\mathcal I_0^{-1}) \, .
\end{aligned}
$$
This is exactly regularity at $\theta_0$ with limit law $\Normal(0,\mathcal I_0^{-1})$.
\end{proof}

\bibliographystyle{IEEEtran}
\bibliography{references}

\end{document}